\documentclass[a4paper,twocolumn,11pt,accepted=2021-12-27]{quantumarticle}
\pdfoutput=1
\usepackage[utf8]{inputenc}
\usepackage[english]{babel}
\usepackage[T1]{fontenc}
\usepackage{amsmath}
\usepackage{hyperref}
\usepackage[numbers,sort&compress]{natbib}
\usepackage{tikz}
\usepackage{lipsum}
\usepackage{amsfonts,amssymb,caption,color,epsfig,graphics,graphicx,hyperref,latexsym,mathrsfs,revsymb,theorem,url,verbatim,enumerate,epstopdf,float,multirow,booktabs,appendix}

\newtheorem{definition}{Definition} 
\newtheorem{proposition}{Proposition} 
\newtheorem{lemma}{Lemma}

\newtheorem{theorem}{Theorem} 
\newtheorem{corollary}[definition]{Corollary}
\newtheorem{conjecture}[definition]{Conjecture}

\newtheorem{remark}[definition]{Remark}
\newtheorem{example}{Example} 
\newtheorem{question}[definition]{Question}

\def\bcj{\begin{conjecture}}
	\def\ecj{\end{conjecture}}
\def\bcr{\begin{corollary}}
	\def\ecr{\end{corollary}}
\def\bd{\begin{definition}}
	\def\ed{\end{definition}}
\def\bea{\begin{eqnarray}}
	\def\eea{\end{eqnarray}}
\def\bem{\begin{enumerate}}
	\def\eem{\end{enumerate}}
\def\bex{\begin{example}}
	\def\eex{\end{example}}
\def\bim{\begin{itemize}}
	\def\eim{\end{itemize}}
\def\bl{\begin{lemma}}
	\def\el{\end{lemma}}
\def\bma{\begin{bmatrix}}
	\def\ema{\end{bmatrix}}
\def\bpf{\begin{proof}}
	\def\epf{\end{proof}}
\def\bpp{\begin{proposition}}
	\def\epp{\end{proposition}}
\def\bqu{\begin{question}}
	\def\equ{\end{question}}
\def\br{\begin{remark}}
	\def\er{\end{remark}}
\def\bt{\begin{theorem}}
	\def\et{\end{theorem}}

\def\squareforqed{\hbox{\rlap{$\sqcap$}$\sqcup$}}
\def\qed{\ifmmode\squareforqed\else{\unskip\nobreak\hfil
		\penalty50\hskip1em\null\nobreak\hfil\squareforqed
		\parfillskip=0pt\finalhyphendemerits=0\endgraf}\fi}
\def\endenv{\ifmmode\;\else{\unskip\nobreak\hfil
		\penalty50\hskip1em\null\nobreak\hfil\;
		\parfillskip=0pt\finalhyphendemerits=0\endgraf}\fi}
\newenvironment{proof}{\noindent \textbf{{Proof.~} }}{\qed}
\def\Dbar{\leavevmode\lower.6ex\hbox to 0pt
	{\hskip-.23ex\accent"16\hss}D}
\makeatletter
\def\url@leostyle{%
	\@ifundefined{selectfont}{\def\UrlFont{\sf}}{\def\UrlFont{\small\ttfamily}}}
\makeatother
\def\bcj{\begin{conjecture}}
	\def\ecj{\end{conjecture}}
\def\bcr{\begin{corollary}}
	\def\ecr{\end{corollary}}
\def\bd{\begin{definition}}
	\def\ed{\end{definition}}
\def\bea{\begin{eqnarray}}
	\def\eea{\end{eqnarray}}
\def\bem{\begin{enumerate}}
	\def\eem{\end{enumerate}}
\def\bex{\begin{example}}
	\def\eex{\end{example}}
\def\bim{\begin{itemize}}
	\def\eim{\end{itemize}}
\def\bl{\begin{lemma}}
	\def\el{\end{lemma}}
\def\bpf{\begin{proof}}
	\def\epf{\end{proof}}
\def\bpp{\begin{proposition}}
	\def\epp{\end{proposition}}
\def\bqu{\begin{question}}
	\def\equ{\end{question}}
\def\br{\begin{remark}}
	\def\er{\end{remark}}
\def\bt{\begin{theorem}}
	\def\et{\end{theorem}}

\def\btb{\begin{tabular}}
	\def\etb{\end{tabular}}

\newcommand{\nc}{\newcommand}

\def\a{\alpha}

\nc{\bbA}{\mathbb{A}} \nc{\bbB}{\mathbb{B}} \nc{\bbC}{\mathbb{C}}
\nc{\bbD}{\mathbb{D}} \nc{\bbE}{\mathbb{E}} \nc{\bbF}{\mathbb{F}}
\nc{\bbG}{\mathbb{G}} \nc{\bbH}{\mathbb{H}} \nc{\bbI}{\mathbb{I}}
\nc{\bbJ}{\mathbb{J}} \nc{\bbK}{\mathbb{K}} \nc{\bbL}{\mathbb{L}}
\nc{\bbM}{\mathbb{M}} \nc{\bbN}{\mathbb{N}} \nc{\bbO}{\mathbb{O}}
\nc{\bbP}{\mathbb{P}} \nc{\bbQ}{\mathbb{Q}} \nc{\bbR}{\mathbb{R}}
\nc{\bbS}{\mathbb{S}} \nc{\bbT}{\mathbb{T}} \nc{\bbU}{\mathbb{U}}
\nc{\bbV}{\mathbb{V}} \nc{\bbW}{\mathbb{W}} \nc{\bbX}{\mathbb{X}}
\nc{\bbZ}{\mathbb{Z}}

\nc{\bA}{{\bf A}} \nc{\bB}{{\bf B}} \nc{\bC}{{\bf C}}
\nc{\bD}{{\bf D}} \nc{\bE}{{\bf E}} \nc{\bF}{{\bf F}}
\nc{\bG}{{\bf G}} \nc{\bH}{{\bf H}} \nc{\bI}{{\bf I}}
\nc{\bJ}{{\bf J}} \nc{\bK}{{\bf K}} \nc{\bL}{{\bf L}}
\nc{\bM}{{\bf M}} \nc{\bN}{{\bf N}} \nc{\bO}{{\bf O}}
\nc{\bP}{{\bf P}} \nc{\bQ}{{\bf Q}} \nc{\bR}{{\bf R}}
\nc{\bS}{{\bf S}} \nc{\bT}{{\bf T}} \nc{\bU}{{\bf U}}
\nc{\bV}{{\bf V}} \nc{\bW}{{\bf W}} \nc{\bX}{{\bf X}}
\nc{\bZ}{{\bf Z}} \nc{\bm}{{\bf m}} \nc{\bv}{{\bf v}}
\nc{\ba}{{\bf a}} \nc{\be}{{\bf e}} \nc{\bu}{{\bf u}}
\nc{\brr}{{\bf r}}

\nc{\cA}{{\cal A}} \nc{\cB}{{\cal B}} \nc{\cC}{{\cal C}}
\nc{\cD}{{\cal D}} \nc{\cE}{{\cal E}} \nc{\cF}{{\cal F}}
\nc{\cG}{{\cal G}} \nc{\cH}{{\cal H}} \nc{\cI}{{\cal I}}
\nc{\cJ}{{\cal J}} \nc{\cK}{{\cal K}} \nc{\cL}{{\cal L}}
\nc{\cM}{{\cal M}} \nc{\cN}{{\cal N}} \nc{\cO}{{\cal O}}
\nc{\cP}{{\cal P}} \nc{\cQ}{{\cal Q}} \nc{\cR}{{\cal R}}
\nc{\cS}{{\cal S}} \nc{\cT}{{\cal T}} \nc{\cU}{{\cal U}}
\nc{\cV}{{\cal V}} \nc{\cW}{{\cal W}} \nc{\cX}{{\cal X}}
\nc{\cZ}{{\cal Z}}

\nc{\hA}{{\hat{A}}} \nc{\hB}{{\hat{B}}} \nc{\hC}{{\hat{C}}}
\nc{\hD}{{\hat{D}}} \nc{\hE}{{\hat{E}}} \nc{\hF}{{\hat{F}}}
\nc{\hG}{{\hat{G}}} \nc{\hH}{{\hat{H}}} \nc{\hI}{{\hat{I}}}
\nc{\hJ}{{\hat{J}}} \nc{\hK}{{\hat{K}}} \nc{\hL}{{\hat{L}}}
\nc{\hM}{{\hat{M}}} \nc{\hN}{{\hat{N}}} \nc{\hO}{{\hat{O}}}
\nc{\hP}{{\hat{P}}} \nc{\hR}{{\hat{R}}} \nc{\hS}{{\hat{S}}}
\nc{\hT}{{\hat{T}}} \nc{\hU}{{\hat{U}}} \nc{\hV}{{\hat{V}}}
\nc{\hW}{{\hat{W}}} \nc{\hX}{{\hat{X}}} \nc{\hZ}{{\hat{Z}}}

\nc{\hn}{{\hat{n}}}

\def\diag{\mathop{\rm diag}}
\def\dim{\mathop{\rm Dim}}

\def\min{\mathop{\rm min}}

\newcommand{\bra}[1]{\langle#1|}
\newcommand{\ket}[1]{|#1\rangle}

\newcommand{\ketbra}[2]{|#1\rangle\!\langle#2|}
\newcommand{\braket}[2]{\langle#1|#2\rangle}

\newcommand{\fl}[2]{\lfloor\frac{#1}{#2}\rfloor}

\def\Dbar{\leavevmode\lower.6ex\hbox to 0pt
	{\hskip-.23ex\accent"16\hss}D}

\begin{document}

\title{Strongly nonlocal unextendible product bases do exist}

	\author{Fei Shi}
	\thanks{Co-first-author}
	\affiliation{School of Cyber Security,
		University of Science and Technology of China, Hefei, 230026, People's Republic of China}
	
	\author{Mao-Sheng Li} 
		\thanks{Co-first-author}
	\affiliation{Department of Physics, Southern University of Science and Technology, Shenzhen 518055, People's Republic of China}
	\affiliation{Department of Physics, University of Science and Technology of China, Hefei 230026, People's Republic of China}
		
	\author{Mengyao Hu}
  \affiliation{LMIB(Beihang University), Ministry of Education, and School of Mathematical Sciences, Beihang University, Beijing 100191, People's Republic of China}

	\author{Lin Chen}
    \affiliation{LMIB(Beihang University), Ministry of Education, and School of Mathematical Sciences, Beihang University, Beijing 100191, People's Republic of China}
    \affiliation{International Research Institute for Multidisciplinary Science, Beihang University, Beijing 100191, People's Republic of China}
	
	\author{Man-Hong Yung}
	\affiliation{Department of Physics, Southern University of Science and Technology, Shenzhen 518055, People's Republic of China}
	\affiliation{Institute for Quantum Science and Engineering, and Department of Physics, Southern University of Science and Technology, Shenzhen, 518055, People's Republic of China}

	\author{Yan-Ling Wang}
\email[]{ wangylmath@yahoo.com}
\affiliation{School of Computer Science and Techonology, Dongguan University of Technology, Dongguan, 523808, People's Republic of China}	
	
	\author{Xiande Zhang}
\email[]{ drzhangx@ustc.edu.cn}
	\affiliation{School of Mathematical Sciences,
		University of Science and Technology of China, Hefei, 230026, People's Republic of China}
		

\maketitle
\small
\begin{abstract}
 	A set of multipartite orthogonal product states is locally irreducible, if it is not possible to eliminate one or more states from the set by orthogonality-preserving local measurements. An effective way to prove that a set is locally irreducible is to show that only trivial orthogonality-preserving local measurement can be performed to this set. In general, it is difficult to show that such an orthogonality-preserving local measurement must be trivial. In this work, we develop two basic techniques to deal with this problem. Using these techniques,    we successfully  show the existence of  unextendible product bases (UPBs) that are  locally irreducible in every bipartition in $d\otimes d\otimes d$ for any $d\geq 3$, and $3\otimes3\otimes 3$ achieves the minimum dimension for the existence of such
	UPBs. These UPBs exhibit the phenomenon of strong quantum nonlocality without entanglement. Our result solves an open question given by Halder \emph{et al.} [\href{https://journals.aps.org/prl/abstract/10.1103/PhysRevLett.122.040403}{Phys. Rev. Lett. \textbf{122}, 040403 (2019)}] and Yuan \emph{et al.} [\href{https://journals.aps.org/pra/abstract/10.1103/PhysRevA.102.042228}{Phys. Rev. A \textbf{102}, 042228 (2020)}].   It also sheds new light on the connections between UPBs and strong quantum nonlocality.
\end{abstract}

\section{Introduction}

Quantum state discrimination has attracted more and more attention in recent years.   Consider a composite quantum system prepared in a state from a known set. The discrimination task is to identify the state by performing measurements on the quantum system.  If it is not possible  to distinguish a set of orthogonal states by performing local operations and classical communications (LOCC),  then this set is said to be locally indistinguishable. The local  indistinguishability has wide applications in data hiding \cite{terhal2001hiding,divincenzo2002quantum,eggeling2002hiding,Matthews2009Distinguishability} and quantum secret sharing \cite{Markham2008Graph}.  Bennett \emph{et al.} firstly gave an example of  locally indistinguishable orthogonal product basis in the bipartite quantum system $3\otimes 3$ \cite{bennett1999quantum}, and it shows the phenomenon of quantum nonlocality without entanglement. Much effort has been devoted to the locally indistinguishable orthogonal product states and orthogonal entangled  states \cite{1,2,3,4,5,6,7,8,9,10,11,12,13,14,15,16,17}. A special class of locally indistinguishable sets called unextendible product bases (UPBs) stood out \cite{bennett1999quantum,de2004distinguishability,1}. A UPB is a set of orthogonal product states whose complementary space contains no product states. UPBs are connected to bound entangled states,  Bell inequalities without quantum violation, and  fermionic systems \cite{bennett1999unextendible,1,Tura2012Four, Chen2014Unextendible, Augusiak2012tight,augusiak2011bell}. Some results of the existence of UPBs with the minimum size were given in \cite{alon2001unextendible,Fen06,Joh13,Chen2013The}, and some explicit UPBs were constructed in \cite{bennett1999quantum,1,Johnston2014The,halder2019family,Agrawal2019Genuinely,Shi2020Unextendible,Bej}.

Recently, Halder \emph{et al.} showed the phenomenon of strong quantum nonlocality without entanglement \cite{Halder2019Strong}.  A set of orthogonal states is strongly nonlocal if it is locally irreducible in every bipartition.  They showed such a phenomenon by presenting two strongly nonlocal orthogonal product bases in $3\otimes 3\otimes 3$ and $4\otimes 4\otimes 4$, respectively. In Refs.~\cite{yuan2020strong,Chen2020,shi2021hyper}, the authors constructed some incomplete strongly nonlocal orthogonal product bases in tripartite systems, but these incomplete product bases are not UPBs.  In Refs.~\cite{2020Strong,li2}, the authors constructed some strongly nonlocal orthogonal entangled sets and bases in tripartite systems. Further, strong quantum nonlocality was  generalized to more general settings \cite{zhangstrong2019,li2020local,Sumit2019Genuinely}. Although many efforts have been  made in the strong quantum nonlocality, the existence of  strongly nonlocal UPBs remains unknown. This is also an open question in Refs.~\cite{Halder2019Strong} and \cite{yuan2020strong}.  It requires more techniques to tackle this problem.

In this paper,  we develop two  useful techniques in Lemmas~\ref{lem:zero} and \ref{lem:trivial}   to prove that a set of  orthogonal product states is strongly nonlocal. By utilizing such techniques, in Propositions~\ref{pro:upb333} and \ref{pro:444stronglyupb}, we successfully show the existence of  strongly nonlocal UPBs  in $3\otimes 3\otimes 3$ and $4\otimes 4\otimes 4$, respectively.  Based on the two strongly nonlocal UPBs, in Theorem~\ref{thm:dddstronglyupb}, we show
that strongly nonlocal UPBs do exist for any tripartite system $d\otimes d\otimes d$ ($d\geq 3$). 

 \section{Preliminaries}\label{sec:pre}
	In this paper, we only consider pure states and
	positive operator-valued measurement (POVM), and we do not normalize states and operators for simplicity.
	
	A set of orthogonal states in $d_1\otimes d_2\otimes \cdots\otimes d_n$ is \emph{locally indistinguishable}, if it is not possible to distinguish the states by using LOCC. A local measurement performed to distinguish a set of multipartite orthogonal states is called an orthogonality-preserving local measurement, if the postmeasurement states remain orthogonal. Further, a set of orthogonal states in $d_1\otimes d_2\otimes \cdots\otimes d_n$ is \emph{locally irreducible} if it is not
	possible to eliminate one or more states from the set by orthogonality-preserving local measurements \cite{Halder2019Strong}. Local irreducibility sufficiently ensures local indistinguishability. However, the converse is not true. For example, consider the following states in $3\otimes 4$,
	\begin{equation}\label{eq:comparing}
		\begin{aligned}
			\ket{\psi_1}&=\ket{0}_A(\ket{0}-\ket{1})_B, \ \ket{\psi_2}=(\ket{0}-\ket{1})_A \ket{2}_B,\\
			\ket{\psi_3}&=\ket{2}_A(\ket{1}-\ket{2})_B, \ \ket{\psi_4}=(\ket{1}-\ket{2})_A\ket{0}_B, \\
			\ket{\psi_5}&=(\ket{0}+\ket{1}+\ket{2})_A(\ket{0}+\ket{1}+\ket{2})_B, \\
			\ket{\psi_6}&=\ket{0}_A\ket{3}_B, \
			\ket{\psi_7}=\ket{1}_A\ket{3}_B,  \   \ket{\psi_8}=\ket{2}_A\ket{3}_B.
		\end{aligned}
	\end{equation}
	Since $\{\ket{\psi_i}\}_{i=1}^{8}$ is a UPB in $3\otimes 4$, $\{\ket{\psi_i}\}_{i=1}^{8}$  is locally  indistinguishable \cite{de2004distinguishability}. However, Bob can perform the measurement  $\{\ket{3}_B\bra{3}, \bbI-\ket{3}_B\bra{3}\}$ to eliminate
	$\{\ket{\psi_i}\}_{i=1}^5$ and $\{\ket{\psi_i}\}_{i=6}^8$, respectively. 
	
	In  $d_1\otimes d_2\otimes\cdots \otimes d_n$, $n\geq 3$ and $d_i\geq 3$, a set of orthogonal product states is \emph{strongly nonlocal} if it is locally irreducible in every bipartition of the subsystems, which shows the phenomenon of strong quantum nonlocality without entanglement \cite{Halder2019Strong}.  The authors in Refs.~\cite{Halder2019Strong,yuan2020strong} also proposed an open question. Whether one can find a strongly nonlocal UPB?  We shall give a positive answer to this open question.
	
	There exists a sufficient condition for showing a set to be locally irreducible. If an  orthogonality-preserving POVM on any of the subsystems is trivial (a measurement is trivial if all
	the POVM elements are proportional to the identity operator), then the set of states is locally irreducible. Throughout this paper, we study the strongest form of nonlocality under the following definition.  A set of orthogonal product states is said to be of  \emph{the strongest nonlocality} if only trivial orthogonality-preserving POVM can be performed for each bipartition of the subsystems. In fact, all known strongly nonlocal sets of tripartite systems  are of the strongest nonlocality. Now  we present two basic lemmas which are   useful for showing  the strongest nonlocality.

\section{Two Basic Lemmas}\label{sec:twolemmas}
	For any positive integer $n\geq 2$, we denote $\bbZ_{n}$ as the set $\{0,1,\cdots,n-1\}$, and let $w_n:=e^{\frac{2\pi \sqrt{-1}}{n}}$.	Let $\cH_n$ be an $n$ dimensional Hilbert space. We always assume that $\{|0\rangle, |1\rangle, \cdots,|n-1\rangle\}$ is the computational basis of $\cH_n$. For any operator $M$ on $\cH_n$, we denote the matrix $M$ as the matrix representation of the operator $M$ under the computational basis. In general, we do not distinguish the operator $M$ and the  matrix $M$. Given any $n\times n$ matrix $E:=\sum_{i=0}^{n-1}\sum_{j=0}^{n-1} a_{i,j}|i\rangle\langle j|$, for $\mathcal{S},\mathcal{T}\subseteq \{|0\rangle,|1\rangle,\cdots, |n-1\rangle\}$, we define
	\begin{equation*}
		{}_\mathcal{S}E_{\mathcal{T}}:=\sum_{|s\rangle \in \mathcal{S}}\sum_{|t\rangle \in \mathcal{T}}a_{s,t} |s\rangle\langle t|.
	\end{equation*}
	It means that ${}_\mathcal{S}E_{\mathcal{T}}$ is a submatrix of $E$ with  row coordinates $\mathcal{S}$ and column coordinates $\mathcal{T}$.  In the case $\cS=\cT$, we denote $E_{\cS}:={}_{\cS}E_{\cS}$ for simplicity. Now, we give two basic lemmas whose proofs will be given in   Appendix~\ref{Appendix:zero_proof}.

	\begin{lemma}[Block Zeros Lemma]\label{lem:zero}
		Let  an  $n\times n$ matrix $E=(a_{i,j})_{i,j\in\bbZ_n}$ be the matrix representation of an operator  $E=M^{\dagger}M$  under the basis  $\cB:=\{\ket{0},\ket{1},\ldots,\ket{n-1}\}$. Given two nonempty disjoint subsets $\cS$ and $\cT$ of $\cB$, assume  that  $\{\ket{\psi_i}\}_{i=0}^{s-1}$, $\{\ket{\phi_j}\}_{j=0}^{t-1}$ are two orthogonal sets  spanned by $\cS$ and $\cT$ respectively, where $s=|\cS|,$ and $t=|\cT|.$  If  $\langle \psi_i| E| \phi_j\rangle =0$
		for any $i\in \mathbb{Z}_s,j\in\mathbb{Z}_t$(we call these zero conditions), then   ${}_\mathcal{S}E_{\mathcal{T}}=\mathbf{0}$  and  ${}_\mathcal{T}E_{\mathcal{S}}=\mathbf{0}$.
	\end{lemma}

	\begin{lemma}[Block Trivial  Lemma]\label{lem:trivial}
		Let  an  $n\times n$ matrix $E=(a_{i,j})_{i,j\in\bbZ_n}$ be the matrix representation of an operator  $E=M^{\dagger}M$  under the basis  $\cB:=\{\ket{0},\ket{1},\ldots,\ket{n-1}\}$. Given a nonempty  subset $\cS:=\{\ket{u_0},\ket{u_1},\ldots,\ket{u_{s-1}}\}$  of $\cB$, let $\{\ket{\psi_j} \}_{j=0}^{s-1}$ be an orthogonal  set spanned by $\cS$.     Assume that $\langle \psi_i|E |\psi_j\rangle=0$ for any $i\neq j\in \mathbb{Z}_s$.  If there exists a state $|u_t\rangle \in\cS$,  such that $ {}_{\{|u_t\rangle\}}E_{\cS\setminus \{|u_t\rangle\}}=\mathbf{0}$   and $\langle u_t|\psi_j\rangle \neq 0$  for any $j\in \mathbb{Z}_s$,   then  $E_{\cS}\propto \mathbb{I}_{\cS}$. (Note that if we consider $\{\ket{\psi_j} \}_{j=0}^{s-1}$ as the Fourier basis, i.e. $\ket{\psi_j}=\sum_{i=0}^{s-1}w_s^{ij}\ket{u_i}$ for  $j\in \mathbb{Z}_s$, then it must have $\langle u_t|\psi_j\rangle \neq 0$  for any $j\in \mathbb{Z}_s$).
	\end{lemma}

\section{UPBs of the strongest  nonlocality in tripartite systems}\label{sec:stronglyUPB}

	All of our UPBs in $d\otimes d\otimes d$ for $d\geq 3$ are from Ref.~ \cite{Agrawal2019Genuinely}. Although there is a statement ``For example in \cite{Halder2019Strong}, the authors
	introduced a new concept called strong quantum nonlocality
	without entanglement. While we have checked that the present UPB does not exhibit such phenomena'' claimed in Sec. VI of Ref.~\cite{Agrawal2019Genuinely}, we will show that these UPBs do exhibit the phenomenon of strong quantum nonlocality without entanglement. In fact, they are of the strongest nonlocality.
	One should note that any of these UPBs  is invariant under  the cyclic permutation of the parties. In order to show that any of these UPBs is of  the strongest nonlocality, we only need to show that   $BC$ party could only perform a trivial orthogonality-preserving POVM by Lemma~\ref{lem:cyc} in Appendix~\ref{Appendix:2lemmas}.
	
	
	It is known that any set of orthogonal product states in $2\otimes n$ is locally distinguishable \cite{1}. A necessary condition for the existence of a  strongly nonlocal UPB in $d_1\otimes d_2\otimes d_3$ is that $d_1,d_2,d_3\geq 3$. Therefore,  the possible minimum quantum system for the existence of such a UPB is the three-qutrit system.   The following set of states is a UPB in $3\otimes 3\otimes 3$ (See Ref. \cite{Agrawal2019Genuinely}).   Let
	
	\begin{equation}\label{eq:upb333}
		\begin{aligned}
			\cA_1&:=\{\ket{\xi_j}_A\ket{0}_B\ket{\eta_i}_C\mid (i,j)\in\bbZ_2\times \bbZ_2\setminus\{(0,0)\}\},\\
			\cA_2&:=\{\ket{\xi_j}_A\ket{\eta_i}_B\ket{2}_C\mid (i,j)\in\bbZ_2\times \bbZ_2\setminus\{(0,0)\}\},\\
			\cA_3&:=\{\ket{2}_A\ket{\xi_j}_B\ket{\eta_i}_C\mid (i,j)\in\bbZ_2\times \bbZ_2\setminus\{(0,0)\}\},\\
			\cB_1&:=\{\ket{\eta_i}_A\ket{2}_B\ket{\xi_j}_C\mid (i,j)\in\bbZ_2\times \bbZ_2\setminus\{(0,0)\}\},\\
			\cB_2&:=\{\ket{\eta_i}_A\ket{\xi_j}_B\ket{0}_C\mid (i,j)\in\bbZ_2\times \bbZ_2\setminus\{(0,0)\}\},\\
			\cB_3&:=\{\ket{0}_A\ket{\eta_i}_B\ket{\xi_j}_C\mid (i,j)\in\bbZ_2\times \bbZ_2\setminus\{(0,0)\}\}, \\
			\ket{S}&=\left(\sum_{i=0}^2\ket{i}\right)_A\left(\sum_{j=0}^2\ket{j}\right)_B\left(\sum_{k=0}^2\ket{k}\right)_C,
		\end{aligned}
	\end{equation}
	where $\ket{\eta_i}=\ket{0}+(-1)^{i}\ket{1}$, $\ket{\xi_j}=\ket{1}+(-1)^{j}\ket{2}$, for $i,j\in\bbZ_2$. We find  that the above UPB is of  the strongest nonlocality.
	
		\begin{figure}[h]
	\centering
	\includegraphics[scale=0.7]{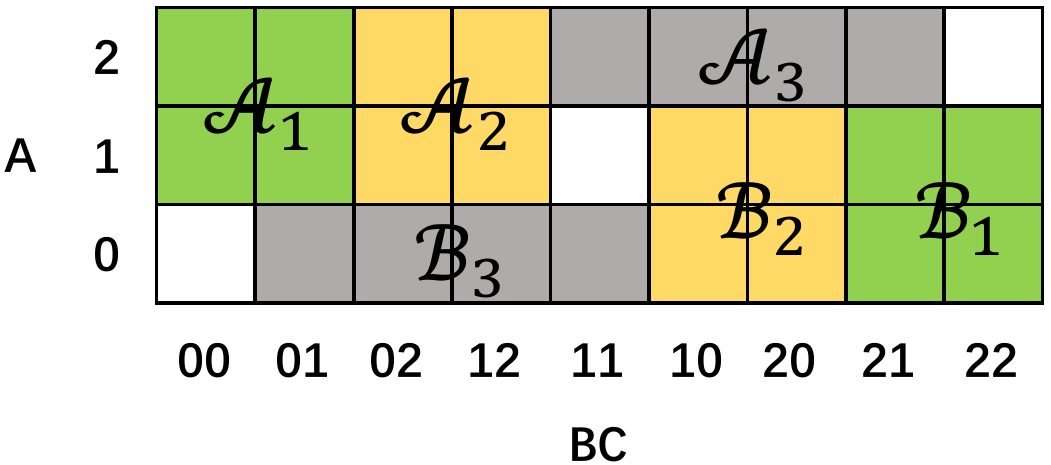}
	\caption{The corresponding $3\times 9$ grid of  $\cup_{i=1}^3\{\cA_i,\cB_i\}$ given by Eq. \eqref{eq:upb333} in $A|BC$ bipartition.  For example, $\cA_1$ corresponds to the $2\times 2$ grid $\{(1,2)\times (00,01)\}$. Moreover,  $\cA_i$ is symmetrical to $\cB_i$ for $1\leq i\leq 3$.   }  \label{fig:tite39}
\end{figure}
	
	\begin{proposition}\label{pro:upb333}
		In $3\otimes 3\otimes 3$, the set $\{\cup_{i=1}^3\{\cA_i,\cB_i\}\cup\{\ket{S}\}\}$ given by Eq.~\eqref{eq:upb333} is a UPB of  the strongest nonlocality. The size of this set is $19$.
	\end{proposition}	
   	\begin{proof}
    	The six subsets $ \cA_i,\cB_i  (i=1,2,3)$ in $A|BC$ bipartition correspond to the six blocks of the $3\times 9$ grid in Fig.~\ref{fig:tite39}.   For example, $\cA_1$ corresponds to the block containing the $2\times 2$ grid $\{(1,2)\times(00,01)\}$ in Fig.~\ref{fig:tite39}.  Moreover,  $\cA_i$ is symmetrical to $\cB_i$ for $1\leq i\leq 3$.  Let $B$ and $C$ come together to  perform a joint  orthogonality-preserving POVM $\{E=M^{\dagger}M\}$, where $E=(a_{ij,k\ell})_{i,j,k,\ell\in \bbZ_3}$. Then the postmeasurement states $\{\bbI\otimes M\ket{\psi}\mid\ket{\psi}\in \{\cup_{i=1}^3\{\cA_i,\cB_i\}\cup\{\ket{S}\}\}\}$ should be mutually orthogonal.

    Assume that $\ket{\psi_1}_A\ket{\psi_2}_B\ket{\psi_3}_C$, $\ket{\varphi_1}_A\ket{\varphi_2}_B\ket{\varphi_3}_C\in  \{\cup_{i=1}^{3}\{\cA_i,\cB_i\} \cup\{\ket{S}\}\}$. Then
   $$ \begin{aligned}
        &{}_A\bra{\psi_1}{}_B\bra{\psi_2}{}_C\bra{\psi_3}\bbI_A\otimes E\ket{\varphi_1}_A\ket{\varphi_2}_B\ket{\varphi_3}_C\\
        =&\braket{\psi_1}{\varphi_1}_A({}_B\bra{\psi_2}{}_C\bra{\psi_3} E\ket{\varphi_2}_B\ket{\varphi_3}_C)=0.
  	\end{aligned}
  	$$
  	If $\braket{\psi_1}{\varphi_1}_A\neq 0$, then ${}_B\bra{\psi_2}{}_C\bra{\psi_3} E\ket{\varphi_2}_B\ket{\varphi_3}_C)=0$.
    By using this property, we need to show that $E\propto \bbI$. 
    			
    First of all, we need to introduce some notations. Let $\mathcal{S}=\{\ket{\psi_1}_A\ket{\psi_2}_B\ket{\psi_3}\}$ be a tripartite orthogonal product set. Define
    \begin{equation*}
    \mathcal{S}(|\psi\rangle_A):=\{ |\psi_2\rangle_B|\psi_3\rangle_C   \mid   |\psi\rangle_A |\psi_2\rangle_B|\psi_3\rangle_C \in  \mathcal{S}\}.
   \end{equation*}
   Moreover, define  $\mathcal{S}^{(A)}=\{\ket{j}_B\ket{k}_C\mid j,k\in\bbZ_n\}$ as the support of $\mathcal{S}(|\psi\rangle_A)$  which spans  $\mathcal{S}(|\psi\rangle_A)$. 
    For example, in Eq.~\eqref{eq:upb333}, $\cA_1:=\{\ket{\xi_j}_A\ket{0}_B\ket{\eta_i}_C\mid (i,j)\in\bbZ_2\times \bbZ_2\setminus\{(0,0)\}  \}$. Then $\cA_1(\ket{\xi_1}_A)=\{\ket{0}_A\ket{\eta_i}_C\}_{i\in \bbZ_2}$, $\cA_1^{(A)}=\{\ket{0}_B\ket{0}_C,\ket{0}_B\ket{1}_C\}$, and $\cA_1(\ket{\xi_1}_A)$ is spanned by $\cA_1^{(A)}$. Actually,  $\{\cA_i^{(A)},\cB_i^{(A)}\}_{i=1}^3$   can be easily observed by Fig.~\ref{fig:tite39}. They are the projection sets of $\{\cA_i,\cB_i\}_{i=1}^3$ in $BC$ party in Fig.~\ref{fig:tite39}.

   	\begin{figure}[t]
   	\centering
   	\includegraphics[scale=0.4]{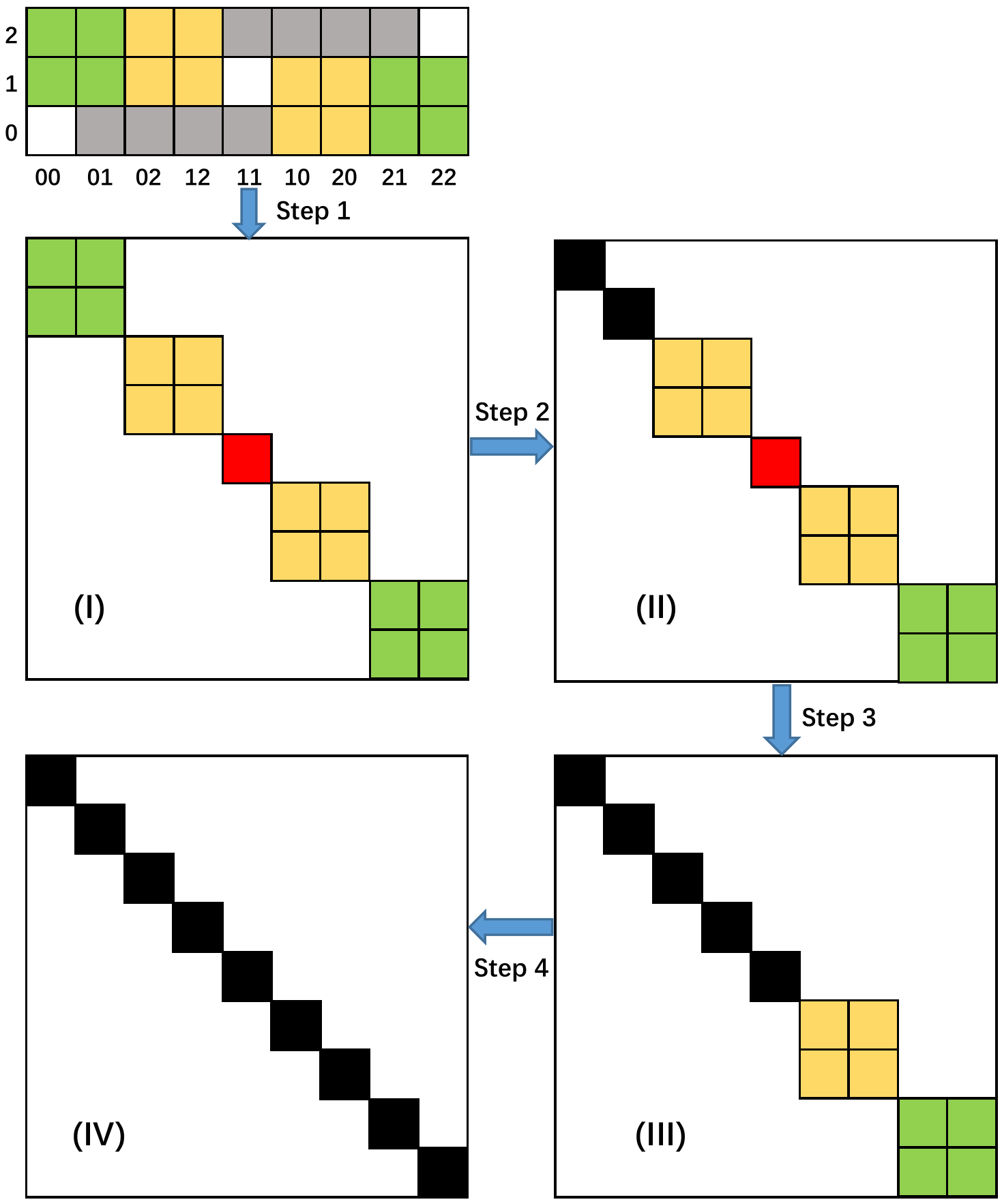}
   	\caption{Proving steps for the strongly nonlocal UPB in $3\otimes 3 \otimes 3$.  }  \label{fig:tite333}
   \end{figure}

   \noindent{\bf Step 1} Since $\ket{\xi_1}_A, \ket{\eta_1}_A$ are non-orthogonal, applying Lemma~\ref{lem:zero} to any two elements of $\{\cA_{1}(\ket{\xi_1}_A),$ $\cA_{2}(\ket{\xi_1}_A),$ $\cB_{2}(\ket{\eta_1}_A),$ $ \cB_{1}(\ket{\eta_1}_A)\}$, we obtain
\begin{equation}\label{eq:four_orthogonal}
\begin{array}{cc}
    {}_{\cA_{i}^{(A)}}E_{\cA_{j}^{(A)}}=\mathbf{0},& {}_{\cA_{i}^{(A)}}E_{\cB_{k}^{(A)}}=\mathbf{0}, \\
     {}_{\cB_{k}^{(A)}}E_{\cB_{\ell}^{(A)}}=\mathbf{0}, & {}_{\cB_{k}^{(A)}}E_{\cA_{i}^{(A)}}=\mathbf{0},
\end{array}
\end{equation}
for $1\leq i\neq j\leq 2$, $1\leq k\neq \ell\leq 2$.
    			
   Next, by using the states $\ket{2}_A\ket{\xi_1}_B\ket{\eta_0}_C\in\cA_3$ and $\{\ket{\xi_1}_A\ket{0}_B\ket{\eta_i}_C\}_{i\in\bbZ_2}\subset\cA_1$, we have
   \begin{equation}
    0={}_B\bra{\xi_1}{}_C\bra{\eta_0}E\ket{0}_B\ket{\eta_i}_C={}_B\bra{1}{}_C\bra{1}E\ket{0}_B\ket{\eta_i}_C,   
   \end{equation}
   	for $ i\in \bbZ_2$. Applying Lemma~\ref{lem:zero} to $\{\ket{1}_B\ket{1}_C\}$ and $\{\ket{0}_B\ket{\eta_i}_C\}_{i\in\bbZ_2}$, we obtain $a_{11,00}=a_{11,01}=0$. In the same way, we can show that $a_{11,02}=a_{11,12}=0$ by using the states $\ket{2}_A\ket{\xi_1}_B\ket{\eta_0}_C\in\cA_3$ and $\{\ket{\xi_1}_A\ket{\eta_i}_B\ket{2}_C\}_{i\in\bbZ_2}\subset\cA_2$. By the symmetry of Fig.~\ref{fig:tite39}, we can also show that $a_{11,10}=a_{11,20}=a_{11,21}=a_{11,22}=0$. Since $E^\dagger=E$,  $E$ is a block diagonal matrix. It can be expressed by
    \begin{equation}\label{eq:matrix333}
   	E=E_{\cA_1^{(A)}}\oplus E_{\cA_2^{(A)}}\oplus E_{\{\ket{1}_B\ket{1}_C\}}\oplus E_{\cB_2^{(A)}}\oplus E_{\cB_1^{(A)}}.
    \end{equation}
   	The intuitive figure of $E$ can be shown in Fig.~\ref{fig:tite333} (I).	
   			
    \noindent{\bf Step 2} By using the states $\{\ket{\xi_1}_A\ket{\eta_i}_B\ket{2}_C\}_{i\in\bbZ_2}\subset\cA_2$, we have ${}_B\bra{\eta_0}{}_C\bra{2} E\ket{\eta_1}_B\ket{2}_C={}_B\bra{\eta_1}{}_C\bra{2} E\ket{\eta_0}_B\ket{2}_C=0$. It implies $a_{02,12}=a_{12,02}$. Moreover, by using the states   $\ket{S}$ and $\{\ket{0}_A\ket{\eta_i}_B\ket{\xi_j}_C\}_{(i,j)\in\bbZ_2\times \bbZ_2\setminus\{(0,0)\}}=\cB_3$, we have
    \begin{equation}
    \left\{\begin{array}{ccl}
    0&=&{}_B\left(\sum_{j=0}^2\langle j|\right){}_C\left(\sum_{k=0}^2 \langle k |\right)E\ket{\eta_1}_B\ket{\xi_1}_C\\[2mm]
    &=&a_{01,01}-a_{02,02}+a_{12,12}-a_{11,11}+a_{00,01},\\[3mm]
    0&=&{}_B\bra{\eta_0}{}_C\bra{\xi_1} E\ket{\eta_1}_B\ket{\xi_0}_C\\[2mm]
    &=&a_{01,01}-a_{02,02}+a_{12,12}-a_{11,11}.
    \end{array}
    \right.
    \end{equation}
    It implies $a_{00,01}=0$.  By using the states $\{\ket{\xi_1}_A\ket{0}_B\ket{\eta_i}_C\}_{i\in\bbZ_2}\subset\cA_1$, we have
    ${}_B\bra{0}{}_C\bra{\eta_0} E\ket{0}_B\ket{\eta_1}_C=0$. It implies $a_{00,00}=a_{01,01}$. Thus
    \begin{equation}\label{eq:eA_1333}
    E_{\cA_1^{(A)}}=k\bbI_{\cA_1^{(A)}}.
    \end{equation}
      The intuitive figure of $E$ can be shown in Fig.~\ref{fig:tite333} (II).				
   
   		\noindent{\bf Step 3}
   Considering    $\{\ket{0}_A\ket{\eta_i}_B\ket{\xi_j}_C\}_{(i,j)\in\mathbb{Z}_2\times\mathbb{Z}_2\setminus}$ ${}_{ \{(0,0)\}}=\cB_3$ and $\ket{S}$. By using Eqs.~\eqref{eq:matrix333}  and \eqref{eq:eA_1333}, we have  the following equality
$$
   \begin{aligned}
   &\sum_{s=0}^1\sum_{t=0}^1 {}_B \bra{s} {}_C\bra{t+1} E \ket{\eta_i}_B\ket{\xi_j}_C\\
   =&\sum_{s=0}^2\sum_{t=0}^2 {}_B \bra{s} {}_C\bra{t} E \ket{\eta_i}_B\ket{\xi_j}_C=0.
   \end{aligned}
$$
   Moreover,  we have
   \begin{equation}
   \sum_{s=0}^1\sum_{t=0}^1\ket{s}_B\ket{t+1}_C=\ket{\eta_0}_B\ket{\xi_0}_C.
   \end{equation}
   Therefore, by using the states $\{\ket{S}\}\cup\{\ket{0}_A\ket{\eta_i}_B\ket{\xi_j}_C\}_{(i,j)\in\bbZ_2\times \bbZ_2\setminus\{(0,0)\}}  $, we have
   \begin{equation}
   {}_B\bra{\eta_k}_C\bra{\xi_\ell}E\ket{\eta_i}_B\ket{\xi_j}_C=0,  
   \end{equation}
 for $(k,\ell)\neq (i,j)\in\bbZ_2\times\bbZ_2$. By Eqs.~\eqref{eq:matrix333}  and \eqref{eq:eA_1333}, we know that  ${}_{\{\ket{0}_B\ket{1}_C\}}E_{\cB_3^{(A)}\setminus \{\ket{0}_B\ket{1}_C\}}=\textbf{0}$. Moreover, ${}_B\bra{0}{}_C\braket{1}{\eta_i}_B\ket{\xi_j}_C\neq 0$ for $(i,j)\in\bbZ_2\times \bbZ_2$.  Applying Lemma~\ref{lem:trivial} to $\{\ket{\eta_i}_B\ket{\xi_j}_C\}_{(i,j)\in\bbZ_2\times \bbZ_2}$, we have
   \begin{equation}\label{eq:B3_333}
   E_{\cB_3^{(A)}}=k_1\bbI_{\cB_3^{(A)}}.
   \end{equation} Since $\cA_1^{(A)}\cap \cB_3^{(A)}\neq \emptyset$, it implies $k=k_1$. Thus, by Eqs.~\eqref{eq:eA_1333}  and \eqref{eq:B3_333}, we obtain
   \begin{equation}
   E_{\cA_1^{(A)}\cup\cB_3^{(A)}}=k\bbI_{\cA_1^{(A)}\cup\cB_3^{(A)}}.
   \end{equation}
    The intuitive figure of $E$ can be shown in Fig.~\ref{fig:tite333} (III).
    
   \noindent{\bf Step 4} 
   By the symmetry of Fig.~\ref{fig:tite416}, we can obtain  $E=k\bbI$.  The intuitive figure of $E$ can be shown in Fig.~\ref{fig:tite333} (IV).
   
    Thus,  $E$  is trivial. This completes the proof.
    \end{proof}
    
   \vspace{0.4cm}	
	The following is a UPB in $4\otimes 4\otimes 4$ (See Ref. \cite{Agrawal2019Genuinely}),
	\begin{equation}\label{eq:444UPB}
		\begin{aligned}
			\cA_0:=&\{\ket{\phi_r}_A\ket{\phi_s}_B\ket{\phi_t}_C\mid (r,s,t)\in\bbZ_2^{(\times 3)}\\ 
			     &\setminus\{(0,0,0)\}\}, \\
			\cA_1:=&\{\ket{\xi_j}_A\ket{0}_B\ket{\eta_i}_C\mid (i,j)\in\bbZ_3\times \bbZ_3\setminus\{(0,0)\}\},\\
			\cA_2:=&\{\ket{\xi_j}_A\ket{\eta_i}_B\ket{3}_C\mid (i,j)\in\bbZ_3\times \bbZ_3\setminus\{(0,0)\}\},\\
			\cA_3:=&\{\ket{3}_A\ket{\xi_j}_B\ket{\eta_i}_C\mid (i,j)\in\bbZ_3\times \bbZ_3\setminus\{(0,0)\}\},\\
			\cB_1:=&\{\ket{\eta_i}_A\ket{3}_B\ket{\xi_j}_C\mid (i,j)\in\bbZ_3\times \bbZ_3\setminus\{(0,0)\}\},\\
			\cB_2:=&\{\ket{\eta_i}_A\ket{\xi_j}_B\ket{0}_C\mid (i,j)\in\bbZ_3\times \bbZ_3\setminus\{(0,0)\}\},\\
			\cB_3:=&\{\ket{0}_A\ket{\eta_i}_B\ket{\xi_j}_C\mid (i,j)\in\bbZ_3\times \bbZ_3\setminus\{(0,0)\}\}, \\
			\ket{S}:=&\left(\sum_{i=0}^3\ket{i}\right)_A\left(\sum_{j=0}^3\ket{j}\right)_B\left(\sum_{k=0}^3\ket{k}\right)_C,
		\end{aligned}
	\end{equation}
	where $\bbZ_2^{(\times 3)}:=\bbZ_2\times\bbZ_2\times \bbZ_2,$
	$\ket{\eta_i}=\sum_{k=0}^2w_3^{ik}\ket{k}$ and $\ket{\xi_j}=\sum_{k=0}^2w_3^{jk}\ket{k+1}$ for $i,j\in\bbZ_3$  and $\ket{\phi_i}=\ket{1}+(-1)^i\ket{2}$ for $i\in\bbZ_2$. Now, we show that the above UPB is of  the strongest nonlocality.

	\begin{proposition}\label{pro:444stronglyupb}
		In $4\otimes 4\otimes 4$, the set $\{\cup_{i=1}^{3}\{\cA_i,\cB_i\}\cup \{\cA_0\} \cup\{\ket{S}\}\}$ given by Eq.~\eqref{eq:444UPB} is a UPB of  the strongest nonlocality. The size of this set is $56$.	
	\end{proposition}

	The proof of Proposition~\ref{pro:444stronglyupb} is given in Appendix~\ref{Appendix:upb444}.  Based on Propositions~\ref{pro:upb333} and \ref{pro:444stronglyupb}, we could  show the existence of UPBs with  the strongest nonlocality in $d\otimes d\otimes d$ for any integer $d\geq 3$.


		Fix an integer $d\geq 3$. For any integer $0\leq k\leq \fl{d-3}{2}$, let $\cC^{(d,d-2k)}=\cup_{i,j=1}^3\{\cA_i^{(d,d-2k)},\cB_i^{(d,d-2k)}\}$, where each $\cA_i^{(d,d-2k)}$ and $\cB_j^{(d,d-2k)}$ are  defined as follows:
		\begin{widetext}
	\begin{equation}\label{eq:upbddd}
		\begin{aligned}
			\cA_1^{(d,d-2k)}&:=\{\ket{\xi_j^{(d-2k)}}_A\ket{k}_B\ket{\eta_i^{(d-2k)}}_C\mid (i,j)\in\bbZ_{d-1-2k}\times \bbZ_{d-1-2k}\setminus\{(0,0)\}\},\\
			\cA_2^{(d,d-2k)}&:=\{\ket{\xi_j^{(d-2k)}}_A\ket{\eta_i^{(d-2k)}}_B\ket{d-1-k}_C\mid (i,j)\in\bbZ_{d-1-2k}\times \bbZ_{d-1-2k}\setminus\{(0,0)\}\},\\
			\cA_3^{(d,d-2k)}&:=\{\ket{d-1-k}_A\ket{\xi_j^{(d-2k)}}_B\ket{\eta_i^{(d-2k)}}_C\mid (i,j)\in\bbZ_{d-1-2k}\times \bbZ_{d-1-2k}\setminus\{(0,0)\}\},\\
			\cB_1^{(d,d-2k)}&:=\{\ket{\eta_i^{(d-2k)}}_A\ket{d-1-k}_B\ket{\xi_j^{(d-2k)}}_C\mid (i,j)\in\bbZ_{d-1-2k}\times \bbZ_{d,d-1-2k}\setminus\{(0,0)\}\},\\
			\cB_2^{(d,d-2k)}&:=\{\ket{\eta_i^{(d-2k)}}_A\ket{\xi_j^{(d-2k)}}_B\ket{k}_C\mid (i,j)\in\bbZ_{d-1-2k}\times \bbZ_{d-1-2k}\setminus\{(0,0)\}\},\\
			\cB_3^{(d,d-2k)}&:=\{\ket{k}_A\ket{\eta_i^{(d-2k)}}_B\ket{\xi_j^{(d-2k)}}_C\mid (i,j)\in\bbZ_{d-1-2k}\times \bbZ_{d-1-2k}\setminus\{(0,0)\}\},
		\end{aligned}
	\end{equation}
	\end{widetext}
	where $\ket{\eta_i^{(d-2k)}}=\sum_{t=k}^{d-2-k}w_{d-1-2k}^{i(t-k)}\ket{t}$, and $\ket{\xi_j^{(d-2k)}}=\sum_{t=k}^{d-2-k}w_{d-1-2k}^{j(t-k)}\ket{t+1}$, for $i,j\in\bbZ_{d-1-2k}$. If $d$ is even, we define  $\cA^{(d,0)}:=\{\ket{\phi_r}_A\ket{\phi_s}_B\ket{\phi_t}_C\mid (r,s,t)\in\bbZ_2\times \bbZ_2\times\bbZ_2\setminus\{(0,0,0)\}\}$
	where  $\ket{\phi_i}=\ket{\frac{d-2}{2}}+(-1)^i\ket{\frac{d}{2}}$ for $i\in\bbZ_2$. Moreover, we define the ``stopper state'' as
	$$
	\ket{S_d}=\left(\sum_{i=0}^{d-1}\ket{i}\right)_A\left(\sum_{j=0}^{d-1}\ket{j}\right)_B\left(\sum_{k=0}^{d-1}\ket{k}\right)_C.$$
	Now, we give a general result.
	
		\begin{theorem}\label{thm:dddstronglyupb}
		In $d\otimes d\otimes d$, $d\geq 3$,
		\begin{enumerate}[(i)]
			\item when $d$ is odd, the set $\{\cup_{k=0}^{\frac{d-3}{2}}\cC^{(d,d-2k)}\cup \{\ket{S_d}\}\}$ given by Eq.~\eqref{eq:upbddd} is a UPB of the strongest nonlocality.  The size of this set is $d^3-4d+4$;	
			\item when $d$ is even, the set $\{\cup_{k=0}^{\frac{d-4}{2}}\cC^{(d,d-2k)}\cup \{\cA^{(d,0)}\}\cup\{\ket{S_d}\}\}$ given by Eq.~\eqref{eq:upbddd} is a UPB of  the strongest nonlocality. The size of this set is $d^3-4d+8$.	
		\end{enumerate}	
	\end{theorem}
	
 The proof of Theorem~\ref{thm:dddstronglyupb} is given in Appendix~\ref{Appendix:upbddd}.	Therefore, we have shown that  strongly nonlocal UPBs in $d\otimes d\otimes d$  do exist for all $d\geq 3$, which answers an open question in Refs.~\cite{Halder2019Strong,yuan2020strong}.  In Ref.~\cite{yuan2020strong}, the authors gave a construction of a strongly nonlocal orthogonal product set of size $6(d-1)^2$ in $d\otimes d\otimes d$ for any $d\geq 3$. When $d=3$, the size of the strongly nonlocal orthogonal product set is $24$, which is bigger than the size of the UPB  given by Eq.~\eqref{eq:upb333}. Note that any their strongly nonlocal orthogonal product set in $d\otimes d\otimes d$ is not a UPB as it is completable.  Our strongly nonlocal UPBs reveal the relationship between strong quantum nonlocality and UPBs. 
 
For the normalized UPB $\{\ket{\psi_i}\}_{i=1}^{t_d}$ in $d\otimes d\otimes d$ of this paper, where $t_d=d^3-8(\fl{d-3}{2}+1)$ for $d\geq 3$, it can be used to construct the bound entangled state 
 \begin{equation}
     \rho=\frac{1}{t_d}\left(\bbI-\sum_{i=1}^{t_d}\ket{\psi_i}\bra{\psi_i}\right)
 \end{equation} 
 with the following new property introduced by a recent result \cite{Hari2021}. 
The bound entangled state $\rho$ has positive partial transposition across every bipartition, but it is not separable across every bipartition \cite{Hari2021}. It is well-known that a separable bipartite state must have positive partial transposition.  Such a bound entangled state shows that  the set of states which are separable across every bipartition is a strict subset of states having positive partial transposition across across every bipartition in $d\otimes d\otimes d$ \cite{Hari2021}. Further, let $\cH_d$ be the subspace spanned by the UPB $\{\ket{\psi_i}\}_{i=1}^{t_d}$, then the complementary space of $\cH_d$ (denote it by $\cH_d^{\bot}$) is an entangled subspace (each state of this subspace is entangled). However, it is not a genuinely entangled subspace (each state of this subspace is genuinely entangled) \cite{Agrawal2019Genuinely}. For example, when $d=3$, the state $\ket{\phi_1}=\ket{\xi_0}_A\ket{0}_B\ket{\eta_0}_C-\ket{\xi_0}_A\ket{\eta_0}_B\ket{2}_C\in \cH_3^{\bot}$, and it is separable across $A|BC$ bipartition. However, if we let $\ket{\phi_2}=\ket{0}_A\ket{\eta_0}_B\ket{\xi_0}_C-\ket{\eta_0}_A\ket{2}_B\ket{\xi_0}_C$ and $\ket{\phi_3}=\ket{2}_A\ket{\xi_0}_B\ket{\eta_0}_C-\ket{\eta_0}_A\ket{\xi_0}_B\ket{0}_C$, then the complementary space of the subspace spanned by  $\{\ket{\psi_i}\}_{i=1}^{t_3}\cup \{|\phi_i\rangle\}_{i=1}^3$ is a genuinely entangled subspace, which is distillable across every bipartition  \cite{Agrawal2019Genuinely}.

 It is possible to apply our techniques to construct strongly nonlocal UPBs in 
$d_1\otimes d_2\otimes d_3$. However, we need to overcome several difficulties. First, we need to construct UPBs in $d_1\otimes d_2\otimes d_3$ (it is possible to generalize the UPB in $d\otimes d\otimes d$ of the paper to the UPB in $d_1\otimes d_2\otimes d_3$). Second, for the strong quantum nonlocality of the UPB, we require that the UPB has a similar structure under cyclic permutation of the subsystems by Lemma~\ref{lem:cyc}. Otherwise, we need to show that any two subsystems can only perform a trivial orthogonality-preserving POVM, and it causes a lot of calculations. Third, for the grid representation of the UPB in $A|BC$ bipartition (see Fig.~\ref{fig:tite39} for an example), we need to use Block Zeros Lemma and Block Trivial Lemma efficiently.   

It is known that UPB is locally indistinguishable \cite{de2004distinguishability}.  One may ask whether there exists a UPB that is locally indistinguishable across every bipartition?    The intuition is to construct a UPB which is still a UPB across every bipartition. Unfortunately, this is a well-known open question \cite{demianowicz2018unextendible}.  Nevertheless, our strongly nonlocal UPB implies the UPB which is locally indistinguishable across every bipartition. Thus our strongly nonlocal UPB solves this problem in a different way.

Finally, we indicate the possible application of the strongly nonlocal UPB in secret sharing. 
Suppose that the information shared by systems Alice, Bob and Charlie is encrypted into some orthogonal quantum states, and the information needs to be revealed together at a later stage. For the common interests, it is assumed that any operation of these participants will not lead to the final failure of the correct disclosure of the information. However, if some subsystems can cooperate, how to ensure the security of information? One finds that for their common interests, the participants can only perform orthogonality preserving measurement. Otherwise, even the global measurement can not reveal the final results.  Under this setting,  the strongest  nonlocality  of the encoded states implies the security of the  encrypted information. In fact,   the maximum success probability  for perfect discrimination of states with the strongest nonlocality is zero without global orthogonality preserving measurements. To perfect discrimination of the strongly nonlocal UPB, there are two methods. First, the three players are collusive, i.e. they can perform a  global orthogonality-preserving POVM, then the strongly nonlocal UPB can be perfectly distinguished \cite{computation2010}.   Second, one can use additional entanglement resources \cite{cohen2008understanding}. For example, let the maximally entangled state  $\ket{\psi}=\sum_{i=0}^{d-1}\ket{i}_A\ket{i}_B$ be shared between Alice and Bob, then Alice  teleports her subsystem to Bob by using the teleportation-based protocol \cite{Bennett1993Teleporting,Sumit2019Genuinely}. Next, let the maximally entangled state  $\ket{\psi}=\sum_{i=0}^{d-1}\ket{i}_B\ket{i}_C$ be shared between Bob and Charlie, then  Charlie  teleports  his subsystem to Bob by using the teleportation-based protocol. Thus, Bob can easily distinguish the strongly nonlocal UPB. It costs $2\log_2d$  ebits of entanglement resource in the above discrimination protocol.  Protocols consuming less entanglement than the teleportation-based protocol attract much attention in recent years \cite{cohen2008understanding,zhang2020locally,Sumit2019Genuinely,2020Strong,Shi2020Unextendible}. 

\section{Conclusion}
\label{sec:con}
We have shown the existence of  strongly nonlocal UPBs in $d\otimes d\otimes d$ for any $d\geq 3$,  which answer an open question in Refs.~\cite{Halder2019Strong} and \cite{yuan2020strong}.   Our results exhibit the relations between quantum nonlocality and UPBs. Recently, Ref.~\cite{shi2021upb} has proved the existence of strongly nonlocal UPBs in general three-,
and four-partite systems. There are some interesting problems left.  Can we construct strongly nonlocal UPBs in $d^{\otimes n}$ for $d\geq 3$ and $n\geq 5$?  Whether there exists a UPB in  $d^{\otimes n}$  which is still a UPB for every bipartition for $d\geq 3$ and  $n\geq 3$?

\section*{Acknowledgments}
\label{sec:ack}	
	The authors are very grateful to the editor and the anonymous  reviewers for providing
many useful suggestions which have greatly improved the presentation of our paper. FS and XZ were supported by the NSFC under Grants No. 11771419 and No. 12171452, the Anhui Initiative in Quantum Information Technologies under Grant No. AHY150200, and the National Key Research and Development Program of China 2020YFA0713100.  MSL and MHY were supported  by  National  Natural  Science  Foundation  of  China  (12005092, 11875160,   and  U1801661), the China Postdoctoral Science Foundation (2020M681996), the  Natural  Science  Foundation  of  Guang-dong  Province  (2017B030308003),   the  Key  R$\&$D  Program of   Guangdong   province   (2018B030326001),   the   Guang-dong    Innovative    and    Entrepreneurial    Research    TeamProgram (2016ZT06D348), the Science, Technology and   Innovation   Commission   of   Shenzhen   Municipality (JCYJ20170412152620376   and   JCYJ20170817105046702 and  KYTDPT20181011104202253),   the Economy, Trade  and  Information  Commission  of  Shenzhen Municipality (201901161512).   YLW is supported by   the NSFC under Grant No. 11901084, 61773119 and the Research startup funds of DGUT under Grant No. GC300501-103.  LC and MH were supported by the  NNSF of China (Grant No. 11871089), and the Fundamental Research Funds for the Central Universities (Grant No. ZG216S2005).

\bibliographystyle{unsrtnat}
\bibliography{reference.bib}







\onecolumn
\appendix

\renewcommand{\theequation}{A\arabic{equation}}
\renewcommand{\theproposition}{A\arabic{proposition}}
\setcounter{equation}{0}
\setcounter{table}{0}
\setcounter{section}{0}
\setcounter{proposition}{0}
\setcounter{lemma}{0}
\setcounter{theorem}{0}
	
\section{The proofs of Lemmas~\ref{lem:zero} and \ref{lem:trivial} }\label{Appendix:zero_proof}
	\subsection{The proof of Lemma~\ref{lem:zero}}
	\begin{lemma}[Block Zeros Lemma]
		Let  an  $n\times n$ matrix $E=(a_{i,j})_{i,j\in\bbZ_n}$ be the matrix representation of an operator  $E=M^{\dagger}M$  under the basis  $\cB:=\{\ket{0},\ket{1},\ldots,\ket{n-1}\}$. Given two nonempty disjoint subsets $\cS$ and $\cT$ of $\cB$, assume  that  $\{\ket{\psi_i}\}_{i=0}^{s-1}$, $\{\ket{\phi_j}\}_{j=0}^{t-1}$ are two orthogonal sets  spanned by $\cS$ and $\cT$ respectively, where $s=|\cS|,$ and $t=|\cT|.$  If  $\langle \psi_i| E| \phi_j\rangle =0$
		for any $i\in \mathbb{Z}_s,j\in\mathbb{Z}_t$(we call these zero conditions), then   ${}_\mathcal{S}E_{\mathcal{T}}=\mathbf{0}$  and  ${}_\mathcal{T}E_{\mathcal{S}}=\mathbf{0}$.
	\end{lemma}

	\begin{proof}
		Since $\{\ket{\psi_i}\}_{i=0}^{s-1}$, $\{\ket{\phi_j}\}_{j=0}^{t-1}$ are two orthogonal sets spanned by $\cS$ and $\cT$ respectively, and $\dim (\text{span} \ \cS)=s$, $\dim (\text{span} \ \cT)=t$, it implies that the subspaces satisfy
		\begin{equation}\label{eq:subspaces}
			\begin{aligned}
				\text{span}\{\ket{\psi_0},\ket{\psi_1},\ldots,\ket{\psi_{s-1}}\}&=	\text{span} \ \cal{S},\\
				\text{span}\{\ket{\phi_0},\ket{\phi_1},\ldots,\ket{\phi_{t-1}}\}&=	\text{span} \ \cal{T}.
			\end{aligned}
		\end{equation}
		For any $|k\rangle \in \cS $ and $|\ell\rangle\in \cT$, by Eq.~\eqref{eq:subspaces}, they are  a linear combination of $\{|\psi_i\rangle\}_{i=0}^{s-1}$  and  $\{|\phi_j\rangle\}_{j=0}^{t-1}$ respectively. Then by the given conditions
		$
		\bra{\psi_i}E\ket{\phi_j}=0 \  (\forall \ i\in \mathbb{Z}_s,j\in\mathbb{Z}_t),
		$
		we would obtain
		\begin{equation}
			a_{k,\ell}=\bra{k}E\ket{\ell}=0.
		\end{equation}
		It means that  ${}_\mathcal{S}E_{\mathcal{T}}=\mathbf{0}$.  Since $E^{\dagger}=E$, we also have ${}_\mathcal{T}E_{\mathcal{S}}=\mathbf{0}$.
	\end{proof}
	\vspace*{0.4cm}	
	
	\subsection{The proof of Lemma~\ref{lem:trivial}}
	\begin{lemma}[Block Trivial  Lemma]
		Let  an  $n\times n$ matrix $E=(a_{i,j})_{i,j\in\bbZ_n}$ be the matrix representation of an operator  $E=M^{\dagger}M$  under the basis  $\cB:=\{\ket{0},\ket{1},\ldots,\ket{n-1}\}$. Given a nonempty  subset $\cS:=\{\ket{u_0},\ket{u_1},\ldots,\ket{u_{s-1}}\}$  of $\cB$, let $\{\ket{\psi_j} \}_{j=0}^{s-1}$ be an orthogonal  set spanned by $\cS$.     Assume that $\langle \psi_i|E |\psi_j\rangle=0$ for any $i\neq j\in \mathbb{Z}_s$.  If there exists a state $|u_t\rangle \in\cS$,  such that $ {}_{\{|u_t\rangle\}}E_{\cS\setminus \{|u_t\rangle\}}=\mathbf{0}$  and $\langle u_t|\psi_j\rangle \neq 0$  for any $j\in \mathbb{Z}_s$, then  $E_{\cS}\propto \mathbb{I}_{\cS}$.  (Note that if we consider $\{\ket{\psi_j} \}_{j=0}^{s-1}$ as the Fourier basis, i.e. $\ket{\psi_j}=\sum_{i=0}^{s-1}w_s^{ij}\ket{u_i}$ for  $j\in \mathbb{Z}_s$, then it must have $\langle u_t|\psi_j\rangle \neq 0$  for any $j\in \mathbb{Z}_s$).
	\end{lemma}
	\begin{proof}
		Without loss of generality,  we can assume that $\ket{u_{i}}=\ket{i}$ for any $i\in\bbZ_s$.  Under this assumption, each of the states $\{\ket{\psi_j}\}_{j=0}^{s-1}$ can be expressed as a linear combination of $\{|i\rangle\}_{i=0}^{s-1}$, i.e.,  $\ket{\psi_j}=\sum_{i=0}^{s-1}h_{i,j}\ket{i}.$  And the set of states $\{\ket{\psi_j}=\sum_{i=0}^{s-1}h_{i,j}\ket{i}\}_{j=0}^{s-1}$ can be normalized as $\{\ket{\varphi_j}=\sum_{i=0}^{s-1} \widetilde{h}_{i,j}\ket{i}\}_{j=0}^{s-1}$.  Then   $H:=(\widetilde{h}_{i,j})_{i,j\in\bbZ_s}$ is an $s\times s$ unitary matrix. Let
		\begin{equation}
			F=\begin{pmatrix}
				H &\textbf{0}_{s\times (n-s)}\\
				\textbf{0}_{(n-s)\times s} &\textbf{0}_{(n-s)\times (n-s)}
			\end{pmatrix}
		\end{equation}
		be an  $n\times n$ matrix. We can define an operator on the space $\cH_n$,
		\begin{equation*}
			F=\sum_{i=0}^{s-1}\sum_{j=0}^{s-1}\widetilde{h}_{i,j}\ketbra{i}{j}.
		\end{equation*}
		Then the matrix $F$ is the matrix representation of the operator $F$ under the basis $\{\ket{0},\ket{1},\ldots,\ket{n-1}\}$.
		
		The set of states $\{\ket{\varphi_j}=\sum_{i=0}^{s-1}\widetilde{h}_{i,j}\ket{i}\}_{j=0}^{s-1}$ can be written as  $\{\ket{\varphi_j}=F\ket{j}\}_{j=0}^{s-1}$. Since $\langle \psi_i|E |\psi_j\rangle=0$ for any $i\neq j\in \mathbb{Z}_s$, it means that  $\langle \varphi_i|E |\varphi_j\rangle=0$ for any $i\neq j\in \mathbb{Z}_s$. Then we have
		\begin{equation}\label{subdiag}
			\bra{i}F^{\dagger}EF\ket{j}=0  \quad \text{for} \ i\neq j\in \mathbb{Z}_s.
		\end{equation}
		 The Eq. \eqref{subdiag} implies that	\begin{equation}
			H^{\dagger}E_{\cS} H=\diag(\alpha_0 \ \alpha_1 \ \cdots \ \alpha_{s-1}),
		\end{equation}
		where $\alpha_i\in \bbC$ for $i\in \bbZ_s$. Since $H$ is a unitary matrix, we have
		\begin{equation}
			E_{\cS} H=H\diag(\alpha_0 \ \alpha_1 \ \cdots \ \alpha_{s-1}).
		\end{equation}
		Since $ {}_{\{|t\rangle\}}E_{\cS\setminus \{|t\rangle\}}=\mathbf{0}$, the $t$-th row of $E_{\cS}$ is $(0 \ 0 \ \cdots \ a_{t,t} \ \cdots \ 0)$.  Then the $t$-th row of $E_{\cS}H$ is  $(a_{t,t}\widetilde{h}_{t,0} \ a_{t,t}\widetilde{h}_{t,1} \ \cdots \ a_{t,t}\widetilde{h}_{t,t}  \ \cdots \ a_{t,t}\widetilde{h}_{t,s-1} )$. Furthermore, the $t$-th row of $H\diag(\alpha_0 \ \alpha_1 \ \cdots \ \alpha_{s-1})$ is $(\alpha_0\widetilde{h}_{t,0} \ \alpha_1\widetilde{h}_{t,1} \ \cdots \ \alpha_{s-1}\widetilde{h}_{t,s-1})$. Then $a_{t,t}\widetilde{h}_{t,j}=\alpha_j\widetilde{h}_{t,j}$ for any $j\in\bbZ_s$.  Since  $\langle t|\psi_j\rangle= h_{t,j}\neq 0$ for any $j\in\bbZ_s$, i.e. $\widetilde{h}_{t,j}\neq 0$ for any $j\in\bbZ_s$, we have $\alpha_j=a_{t,t}$ for any $j\in\bbZ_s$.  Then
		\begin{equation}
			E_{\cS}=H\diag(a_{t,t} \ \a_{t,t}\ \cdots \ a_{t,t})H^{\dagger}=\diag(a_{t,t} \ \a_{t,t}\ \cdots \ a_{t,t}).
		\end{equation}
		Thus, $E_{\cS}\propto \mathbb{I}_{\cS}.$
	\end{proof}
	\vspace*{0.4cm}

	\section{Two more lemmas used in this paper}\label{Appendix:2lemmas}

	\begin{lemma}\label{lem:cyc}
		Let $\{\ket{\psi}\}\subset\otimes_{i=1}^{n}\cH_{A_i}$ be a set of orthogonal states. Define  $B_1=\{A_2A_3\ldots A_n\}$, $B_2=\{A_3\ldots A_nA_1\}, B_3=\{A_4\ldots A_nA_1A_2\}, \ldots, B_n=\{A_1\ldots A_{n-1}\}$. If $B_i$ party can only perform a trivial orthogonality-preserving POVM for any $1\leq i\leq n$,  then the set  $\{\ket{\psi}\}$  is of the strongest nonlocality.
	\end{lemma}
	\begin{proof}
		For any nontrivial bipartition $A_{i_1}\ldots A_{i_j}| A_{i_{j+1}} \ldots A_{i_n}$ of the subsystems, where $(i_1,i_2,\cdots,i_n)$ is a permutation of $(1,2,\cdots,n)$ and $1\leq j\leq n-1$.
		There exist some $r,s\in\{1,2,\cdots,n\}$, such that $A_{i_1}\ldots A_{i_j}\subset B_r$ and $A_{i_{j+1}} \ldots A_{i_n}\subset B_s$. Therefore,  both parties $A_{i_1}\ldots A_{i_j}$ and $A_{i_{j+1}} \ldots A_{i_n}$ can only perform a trivial orthogonality-preserving POVM.
	\end{proof}
	\vspace*{0.4cm}

	\begin{lemma}\label{lem:nonorthogonal}
		Let $\ket{\eta_1^{(d-2k)}}=\sum_{t=k}^{d-2-k}w_{d-1-2k}^{t-k}\ket{t}$, $\ket{\xi_1^{(d-2k)}}=\sum_{t=k}^{d-2-k}w_{d-1-2k}^{t-k}\ket{t+1}$. Then  	$\braket{\xi_1^{(d-2\ell_1)}}{\xi_1^{(d-2\ell_2)}}\neq 0$, 	$\braket{\eta_1^{(d-2\ell_3)}}{\eta_1^{(d-2\ell_4)}}\neq 0$, and $\braket{\xi_1^{(d-2\ell_5)}}{\eta_1^{(d-2\ell_6)}}\neq 0$, where $d-2\ell_i\geq 3$ for $1\leq i\leq 6$.
	\end{lemma}
	\begin{proof}
		Let  $w_{d_1}=e^{\frac{2\pi i}{d_1}}$,  $w_{d_2}=e^{\frac{2\pi i}{d_2}}$, where $d_1,d_2$ are positive integers. We claim that if $0<k\leq \min\{d_1,d_2\}$, then	$\sum_{j=0}^{k-1}\overline{w_{d_1}^j}w_{d_2}^j\neq 0$.
		
		If $d_1=d_2$, then 	$\sum_{j=0}^{k-1}\overline{w_{d_1}^j}w_{d_2}^j=k\neq 0$. If $d_1\neq d_2$, $\frac{k}{d_1},$ and $\frac{k}{d_2}$ are two different elements lying in the interval $(0,1]$.  Therefore, $\frac{k}{d_2}-\frac{k}{d_1}$ cannot be an integer. Then we can obtain that  $(\overline{w_{d_1}}w_{d_2})^k=e^{(\frac{k}{d_2}-\frac{k}{d_1})2\pi i}\neq 1$. Therefore,
		\begin{equation}
			\sum_{j=0}^{k-1}\overline{w_{d_1}^j}w_{d_2}^j=	\sum_{j=0}^{k-1}(\overline{w_{d_1}}w_{d_2})^j=\frac{1-(\overline{w_{d_1}}w_{d_2})^k}{1-\overline{w_{d_1}}w_{d_2}}\neq 0.
		\end{equation}
		
		Without loss of generality, we assume that $\ell_1\leq \ell_2$ and $\ell_3\leq \ell_4$.
		\begin{equation*}
			\begin{aligned}
				\braket{\xi_1^{(d-2\ell_1)}}{\xi_1^{(d-2\ell_2)}}&=w_{d-1-2\ell_1}^{\ell_1-\ell_2}\sum_{j=0}^{d-2-2\ell_2}\overline{w_{d-1-2\ell_1}^j}w_{d-1-2\ell_2}^j\neq 0,\\
				\braket{\eta_1^{(d-2\ell_3)}}{\eta_1^{(d-2\ell_4)}}&=w_{d-1-2\ell_3}^{\ell_3-\ell_4}\sum_{j=0}^{d-2-2\ell_4}\overline{w_{d-1-2\ell_3}^j}w_{d-1-2\ell_4}^j\neq 0,
			\end{aligned}
		\end{equation*}
		by the above claim. If  $\ell_5\leq \ell_6-1$, then 	\begin{equation*}
			\braket{\xi_1^{(d-2\ell_5)}}{\eta_1^{(d-2\ell_6)}}=w_{d-1-2\ell_5}^{\ell_5+1-\ell_6}\sum_{j=0}^{d-2-2\ell_6}\overline{w_{d-1-2\ell_5}^j}w_{d-1-2\ell_6}^j\neq 0.
		\end{equation*}
		If  $\ell_5= \ell_6$, then 	
		\begin{equation*}
			\braket{\xi_1^{(d-2\ell_5)}}{\eta_1^{(d-2\ell_6)}}=w_{d-1-2\ell_6}^{\ell_5+1-\ell_6}\sum_{j=0}^{d-3-\ell_6-\ell_5}\overline{w_{d-1-2\ell_5}^j}w_{d-1-2\ell_6}^j\neq 0.
		\end{equation*}	
		If  $\ell_5\geq \ell_6+1$, then 	
		\begin{equation*}
			\braket{\xi_1^{(d-2\ell_5)}}{\eta_1^{(d-2\ell_6)}}=w_{d-1-2\ell_6}^{\ell_5+1-\ell_6}\sum_{j=0}^{d-2-2\ell_5}\overline{w_{d-1-2\ell_5}^j}w_{d-1-2\ell_6}^j\neq 0.
		\end{equation*}	
	\end{proof}

	\vspace*{0.4cm}

	\section{The proof of Proposition~\ref{pro:444stronglyupb}}\label{Appendix:upb444}
	
	\begin{figure}[h]
		\centering
		\includegraphics[scale=0.9]{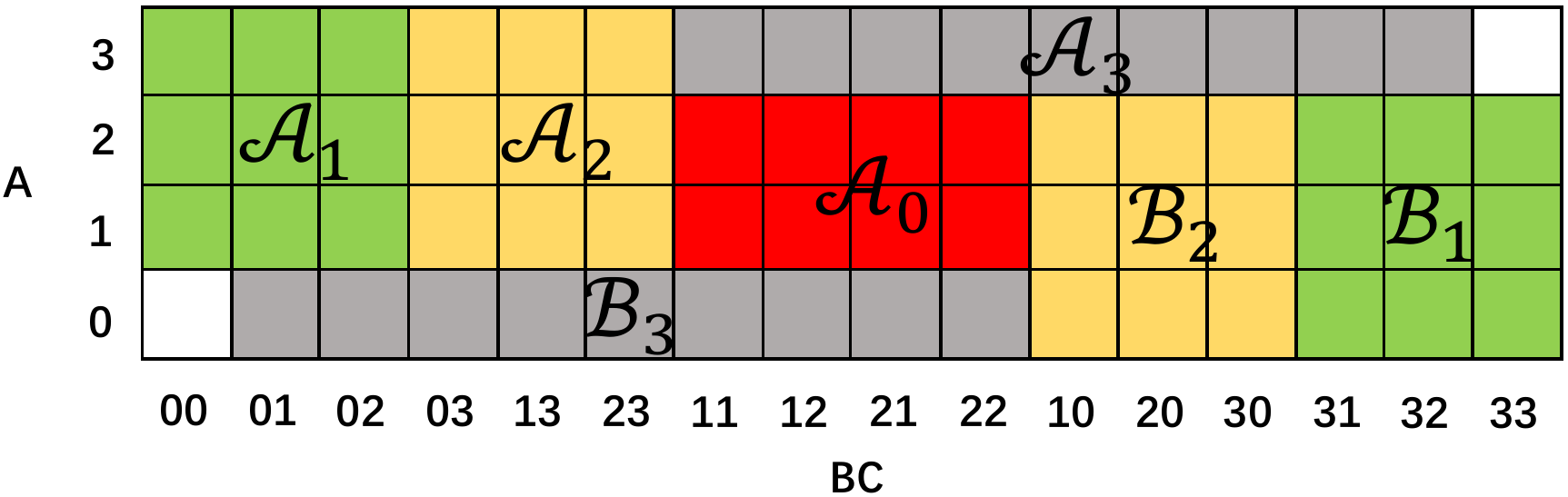}
		\caption{The corresponding $4\times 16$ grid of  $\cup_{i=1}^{3}\{\cA_i,\cB_i\}\cup \{\cA_0\}$ in $A|BC$ bipartition.  For example, $\cA_1$ correspond to the $3\times 3$ grid $\{(1,2,3)\times (00,01,02)\}$.  }  \label{fig:tite416}
	\end{figure}
	\begin{proof}
		The seven subsets  $\cA_0, \cA_i,\cB_i (i=1,2,3)   $ in $A|BC$ bipartition correspond to the seven blocks of $4\times 16$ grid in Fig.~\ref{fig:tite416}.
		Let $B$ and $C$ come together to  perform a joint  orthogonality-preserving POVM $\{E=M^{\dagger}M\}$, where $E=(a_{ij,k\ell})_{i,j,k,\ell\in\bbZ_4}$. Then the postmeasurement states $\{\bbI\otimes M\ket{\psi}\mid\ket{\psi}\in \{\cup_{i=1}^{3}\{\cA_i,\cB_i\}\cup \{\cA_0\} \cup\{\ket{S}\}\} \}$ should be mutually orthogonal.
		
		\noindent{\bf Step 1} Since $\ket{\xi_1}_A, \ket{\phi_1}_A, \ket{\eta_1}_A$ are mutually non-orthogonal, applying Lemma~~\ref{lem:zero} to any two elements of $\{\cA_{1}(\ket{\xi_1}_A),\cA_{2}(\ket{\xi_1}_A),\cA_{0}(\ket{\phi_1}_A),\cB_{2}(\ket{\eta_1}_A),\cB_{1}(\ket{\eta_1}_A)\}$, we obtain that $E$ is a block diagonal matrix,
		\begin{equation}\label{eq:matrixM}
			E=E_{\cA_1^{(A)}}\oplus E_{\cA_2^{(A)}}\oplus E_{\cA_0^{(A)}}\oplus E_{\cB_2^{(A)}}\oplus E_{\cB_1^{(A)}}.
		\end{equation}
			The intuitive figure of $E$ can be shown in Fig.~\ref{fig:tite444} (I).
			
		   	\begin{figure}[t]
			\centering
			\includegraphics[scale=0.5]{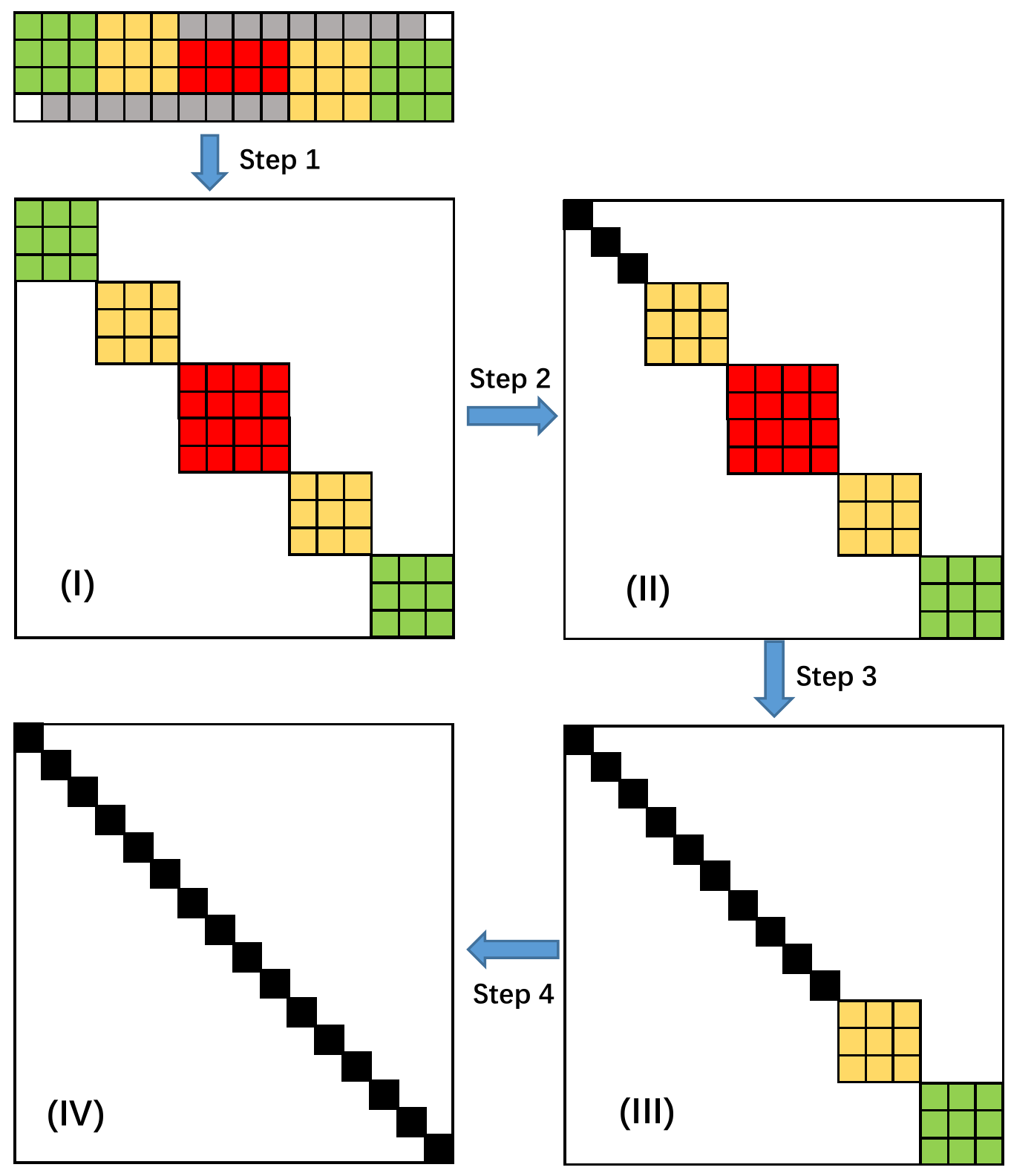}
			\caption{Proving steps for the strongly nonlocal UPB in $4\otimes 4 \otimes 4$.  }  \label{fig:tite444}
		\end{figure}
	
		\noindent{\bf Step 2} By using the states in $\{\ket{\xi_1}_A\ket{0}_B\ket{\eta_i}_C\}_{i\in\bbZ_3}\subset \cA_1$, we have
		\begin{equation}\label{eq:operatorMu}
			{}_B\bra{0}{}_C\bra{\eta_i}E\ket{0}_B\ket{\eta_j}_C=0, \quad \text{for} \  i\neq j\in\bbZ_3.
		\end{equation}
		Then there exist    real numbers $a_s$ for all $s\in\bbZ_3$  such that
		\begin{equation}
			E_{\cA_1^{(A)}}=\sum_{s=0}^2a_s\ket{0}_B\bra{0}\otimes\ket{\eta_s}_C\bra{\eta_s},
		\end{equation}
		as $E=E^{\dagger}$. In the same way, there exist real numbers $a_s,b_s,c_t,d_t,e_{s,t}$ such that the operator
		\begin{equation}\label{eq:operatorM}
			\begin{aligned}
				E=&\sum_{s=0}^2a_s\ket{0}_B\bra{0}\otimes\ket{\eta_{s}}_C\bra{\eta_{s}}+\sum_{s=0}^2 b_s\ket{\eta_{s}}_B\bra{\eta_{s}}\otimes\ket{3}_C\bra{3}+\sum_{t=0}^2 c_t\ket{\xi_{t}}_B\bra{\xi_{t}}\otimes\ket{0}_C\bra{0}\\
				&+\sum_{t=0}^2 d_t\ket{3}_B\bra{3}\otimes\ket{\xi_{t}}_C\bra{\xi_{t}}+\sum_{s=0}^1\sum_{t=0}^1 e_{s,t}\ket{\phi_s}_B\bra{\phi_s}\otimes\ket{\phi_t}_C\bra{\phi_t}.
			\end{aligned}
		\end{equation}	
		By using 	those states $ \{\ket{0}_A\ket{\eta_i}_B\ket{\xi_j}_C\}_{(i,j)\in\mathbb{Z}_3\times \mathbb{Z}_3\setminus \{(0,0)\}}=\cB_3$, we can show that
		\begin{equation}\label{theOrthogonal}
			{}_B\bra{\eta_k} {}_C\bra{\xi_\ell} E \ket{\eta_i}_B\ket{\xi_j}_C=0,  \quad \text{for} \ (k,\ell)\neq (i,j) \in\mathbb{Z}_3\times \mathbb{Z}_3\setminus \{(0,0)\}.
		\end{equation}
		Assume $k\neq i$. By Eq.~\eqref{eq:operatorM}, we have
		\begin{equation}\label{eq:thezero}
			\begin{aligned}
				0=&{}_B\bra{\eta_k} {}_C\bra{\xi_\ell}  E \ket{\eta_i}_B\ket{\xi_j}_C\\
				=&\sum_{s=0}^2a_s   \braket{ \eta_k}{0}_B \braket{0}{\eta_i}_B  \braket{\xi_\ell}{\eta_s}_C \braket{\eta_s}{\xi_j}_C + \sum_{s=0}^2b_s  \braket{\eta_k}{\eta_s}_B\braket{\eta_s}{\eta_i}_B\braket{\xi_\ell}{3}_C\braket{3}{\xi_j}_C\\
				&+\sum_{s=0}^1\sum_{t=0}^1e_{s,t}\braket{\eta_k}{\phi_s}_B\braket{\phi_s}{\eta_i}_B\braket{\xi_\ell}{\phi_t}_C\braket{\phi_t}{\xi_j}_C\\
				=&\sum_{s=0}^2a_s\braket{\xi_\ell}{\eta_s}_C \braket{\eta_s}{\xi_j}_C+\sum_{s=0}^1\sum_{t=0}^1e_{s,t}\braket{\eta_k}{\phi_s}_B\braket{\phi_s}{\eta_i}_B\braket{\xi_\ell}{\phi_t}_C\braket{\phi_t}{\xi_j}_C.
			\end{aligned}
		\end{equation}
		If $k=0$, $\ell=i=j=1$, and   $k=0$, $\ell=2$, $i=1$, $j=2$, then by Eq.~\eqref{eq:thezero}, we have
		\begin{equation}
			\left\{
			\begin{aligned}
				&a_0+4a_1+a_2-2e_{0,0}-2(1-w^2)(1-w)e_{0,1}=0,\\
				&a_0+a_1+4a_2-2e_{0,0}-2(1-w)(1-w^2)e_{0,1}=0.
			\end{aligned}\right.
		\end{equation}
		It implies $a_1=a_2$. Moreover, let $k=0$,  $\ell=1$, $i=1$, $j=0$, by Eq.~\eqref{eq:thezero}, we have
		\begin{equation}\label{eq:a0=a1}
			2a_0+2a_1-a_2-4e_{0,0}=0.
		\end{equation}
		Next, by using the states $\ket{S}$, $\ket{0}_A\ket{\eta_1}_B\ket{\xi_0}_C\in\cB_3$ and Eq.~\eqref{eq:operatorM}, we have
		\begin{equation}\label{eq:a1a2}
			0={}_B(\sum_{j=0}^3\langle j|){}_C(\sum_{k=0}^3 \langle k |)E\ket{\eta_1}_B\ket{\xi_0}_C=6a_0-8e_{0,0}.
		\end{equation}
		Then by Eqs.~\eqref{eq:a0=a1} and \eqref{eq:a1a2} and $a_1=a_2$, we would obtain  $a_0=a_1$. Thus $a_0=a_1=a_2$. It means that the operator
		\begin{equation}
			E_{\cA_1^{(A)}}\propto\sum_{s=0}^2\ket{0}_B\bra{0}\otimes\ket{\eta_s}_C\bra{\eta_s},
		\end{equation}
		which is equivalent to 
		\begin{equation}\label{eq:eA_1}
		E_{\cA_1^{(A)}}=k \bbI_{\cA_1^{(A)}}.
		\end{equation}
			The intuitive figure of $E$ can be shown in Fig.~\ref{fig:tite444} (II).

		\noindent{\bf Step 3}
		Considering $\ket{S}$ and $\{\ket{0}_A\ket{\eta_i}_B\ket{\xi_j}_C\}_{(i,j)\in\mathbb{Z}_3\times \mathbb{Z}_3\setminus \{(0,0)\}}=\cB_3$. By using Eqs.~\eqref{eq:matrixM}  and \eqref{eq:eA_1}, we have  the following equality
		\begin{equation}
			\sum_{s=0}^2\sum_{t=0}^2 {}_B \bra{s} {}_C\bra{t+1} E \ket{\eta_i}_B\ket{\xi_j}_C=\sum_{s=0}^3\sum_{t=0}^3 {}_B \bra{s} {}_C\bra{t} E \ket{\eta_i}_B\ket{\xi_j}_C=0.
		\end{equation}
		Moreover,  we have
		\begin{equation}
			\sum_{s=0}^2\sum_{t=0}^2\ket{s}_B\ket{t+1}_C=\ket{\eta_0}_B\ket{\xi_0}_C.
		\end{equation}
		Therefore, by using the states $\{\ket{S}\}\cup\{\ket{0}_A\ket{\eta_i}_B\ket{\xi_j}_C\}_{(i,j)\in\bbZ_3\times \bbZ_3\setminus\{(0,0)\}}$, we have
		\begin{equation}
			{}_B\bra{\eta_k}_C\bra{\xi_\ell}E\ket{\eta_i}_B\ket{\xi_j}_C=0,  \quad  \text{for} \ (k,\ell)\neq (i,j)\in\bbZ_3\times\bbZ_3.
		\end{equation}
		
		For any $\ket{t_1}_B\ket{t_2}_C\in \cA_1^{(A)}\cap\cB_3^{(A)}$, we have  ${}_{\{\ket{t_1}_B\ket{t_2}_C\}}E_{\cB_3^{(A)}\setminus \{\ket{t_1}_B\ket{t_2}_C\}}=\textbf{0}$ by Eqs.~\eqref{eq:matrixM}  and \eqref{eq:eA_1}. Moreover, ${}_B\bra{t_1}{}_C\braket{t_2}{\eta_i}_B\ket{\xi_j}_C\neq 0$ for $(i,j)\in\bbZ_3\times \bbZ_3$.  Applying Lemma~\ref{lem:trivial} to $\{\ket{\eta_i}_B\ket{\xi_j}_C\}_{(i,j)\in\bbZ_3\times \bbZ_3}$, we have
		\begin{equation}\label{eq:B3}
		E_{\cB_3^{(A)}}=k_1\bbI_{\cB_3^{(A)}}.
		\end{equation} Since $\cA_1^{(A)}\cap \cB_3^{(A)}\neq \emptyset$, it implies $k=k_1$. Thus, by Eqs.~\eqref{eq:eA_1}  and \eqref{eq:B3}, we obtain
		\begin{equation}
		E_{\cA_1^{(A)}\cup\cB_3^{(A)}}=k\bbI_{\cA_1^{(A)}\cup\cB_3^{(A)}}.
		\end{equation}
			The intuitive figure of $E$ can be shown in Fig.~\ref{fig:tite444} (III).
		
		\noindent{\bf Step 4}
		By the symmetry of Fig.~\ref{fig:tite416}, we can obtain  $E=k\bbI$. 			The intuitive figure of $E$ can be shown in Fig.~\ref{fig:tite444} (IV). 
		
		Thus,  $E$  is trivial. This completes the proof.
	\end{proof}
	\vspace{0.4cm}
	
	\section{The proof of Theorem~\ref{thm:dddstronglyupb}}\label{Appendix:upbddd}

		\begin{figure}[h]
		\centering
		\includegraphics[scale=0.5]{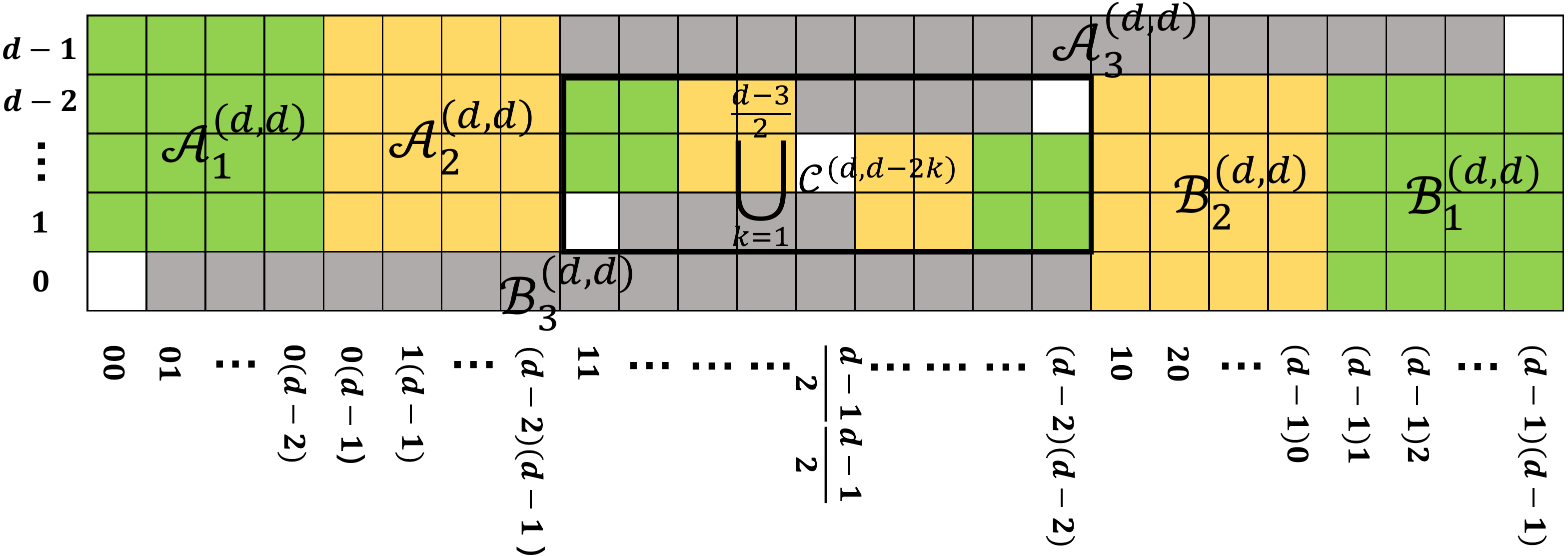}
		\caption{The corresponding $d\times d^2$ grid of $\cup_{k=0}^{\frac{d-3}{2}}\cC^{(d,d-2k)}$ when $d\geq 3$ is odd. }  \label{fig:titedodd}
	\end{figure}

	\begin{proof}
		(i)	We prove it by induction on $d$. We have shown that the conclusion holds when $d=3$ by Proposition \ref{pro:upb333}. Denote $\delta=\frac{d-1}{2}$. Assume  that $\{\cup_{k=0}^{\delta-2}\cC^{(d-2,(d-2)-2k)}\cup \{\ket{S_{d-2}}\}\}$   is  of the strongest nonlocality when $d-2\geq 3$. We need to show that  $\{\cup_{k=0}^{\delta-1}\cC^{(d,d-2k)}\cup \{\ket{S_d}\}\}$  is also of  the strongest nonlocality. Define a bijection, $\ket{j}\rightarrow \ket{j+1}$, then $$\{\cup_{k=0}^{\delta-2}\cC^{(d-2,(d-2)-2k)}\cup \{\ket{S_{d-2}}\}\}\rightarrow \{\cup_{k=1}^{\delta-1}\cC^{(d,d-2k)}\cup \{\ket{S_{d-2}'}\}\},$$ where $\ket{S_{d-2}'}=(\sum_{i=1}^{d-2}\ket{i})_A(\sum_{j=1}^{d-2}\ket{j})_B(\sum_{k=1}^{d-2}\ket{k})_C$. With this bijection, the set $\{\cup_{k=1}^{\delta-1}\cC^{(d,d-2k)}\cup \{\ket{S_{d-2}'}\}\}$ is also of the strongest nonlocality with respect to its domain subspace.
		
		The $6\delta$ subsets $ \cC^{(d,d-2k)}$ ($ k=0,1,\cdots,\delta-1,$    and   each $\cC^{(d,d-2k)}$ contains $6$ subsets) correspond to the blocks of the $d\times d^2$ grid in Fig.~\ref{fig:titedodd}.  Let $B$ and $C$ come together and perform the orthogonality-preserving joint POVM $\{E=M^{\dagger}M\}$, where $E=(a_{ij,k\ell})_{ i,j,k,\ell\in \bbZ_d}$. Then the postmeasurement states $\{\bbI\otimes M\ket{\psi}\mid \ket{\psi}\in \{\cup_{k=0}^{\delta-1}\cC^{(d,d-2k)}\cup \{\ket{S_d}\}\} \}$ should be mutually orthogonal.
		
		\noindent{\bf Step 1}  Since the states $\{\ket{\xi_1^{(d-2k)}}_A,\ket{\eta_1^{(d-2k)}}_A\}_{k=0}^{\delta-1}$ are mutually non-orthogonal by Lemma~\ref{lem:nonorthogonal} in Appendix~\ref{Appendix:2lemmas},  applying Lemma~\ref{lem:zero}  to any two  elements of $\{\cA_1^{(d,d-2k)}(\ket{\xi_1^{(d-2k)}}_A),$ $\cA_2^{(d,d-2k)}(\ket{\xi_1^{(d-2k)}}_A),$ $ \cB_2^{(d,d-2k)}(\ket{\eta_1^{(d-2k)}}_A),$ $ \cB_1^{(d,d-2k)}(\ket{\eta_1^{(d-2k)}}_A) \}_{k=0}^{\delta-1}$, we obtain that $E$ is a block diagonal matrix,
		except  $a_{\delta\delta,k\ell}$, $a_{ij,\delta\delta}$ for $0\leq i,j,k,\ell\leq d-1$. By using the states  $\{\ket{\xi_1^{(d)}}_A\ket{0}_B\ket{\eta_i^{(d)}}_C\}_{i\in\bbZ_{d-1}}\subset\cA_1^{(d,d)}$ and $\ket{d-1}_A\ket{\xi_1^{(d)}}_B\ket{\eta_1^{(d)}}_C\in\cA_3^{(d,d)}$, we have $a_{0j,\delta\delta}=0$ for $0\leq j\leq d-2$ by Lemma~\ref{lem:zero}. In the same way, we can obtain that $a_{i(d-1),\delta\delta}=0$ for $0\leq i\leq d-2$ by using the states $\{\ket{\xi_1^{(d)}}_A\ket{\eta_i^{(d)}}_B\ket{d-1}_C\}_{i=0}^{d-2}\subset\cA_2^{(d,d)}$ and $\ket{d-1}_A\ket{\xi_1^{(d)}}_B\ket{\eta_1^{(d)}}_C\in\cA_3^{(d,d)}$. By the symmetry of Fig.~\ref{fig:titedodd}, we can also show that $a_{i0,\delta\delta}=0$ for $1\leq i\leq d-1$, and $a_{(d-1)j,\delta\delta}=0$ for $1\leq j\leq d-1$.  Since $E^{\dagger}=E$, $E$ is a block diagonal matrix. It can be expressed by
		\begin{equation}\label{eq:matrix}
			E=E_{(\cA_1^{(d,d)})^{(A)}}\oplus E_{(\cA_2^{(d,d)})^{(A)}}\oplus E_{\cC}\oplus E_{(\cB_2^{(d,d)})^{(A)}}\oplus E_{(\cB_1^{(d,d)})^{(A)}},
		\end{equation}
		where $E_{\cC}:=E_{\{\cup_{k=1}^{\delta-1}\cC^{(d,d-2k)}\}^{(A)}}$.
		The intuitive figure of $E$ can be shown in Fig.~\ref{fig:titeddd} (I).
		
		\begin{figure}[t]
		\centering
		\includegraphics[scale=0.7]{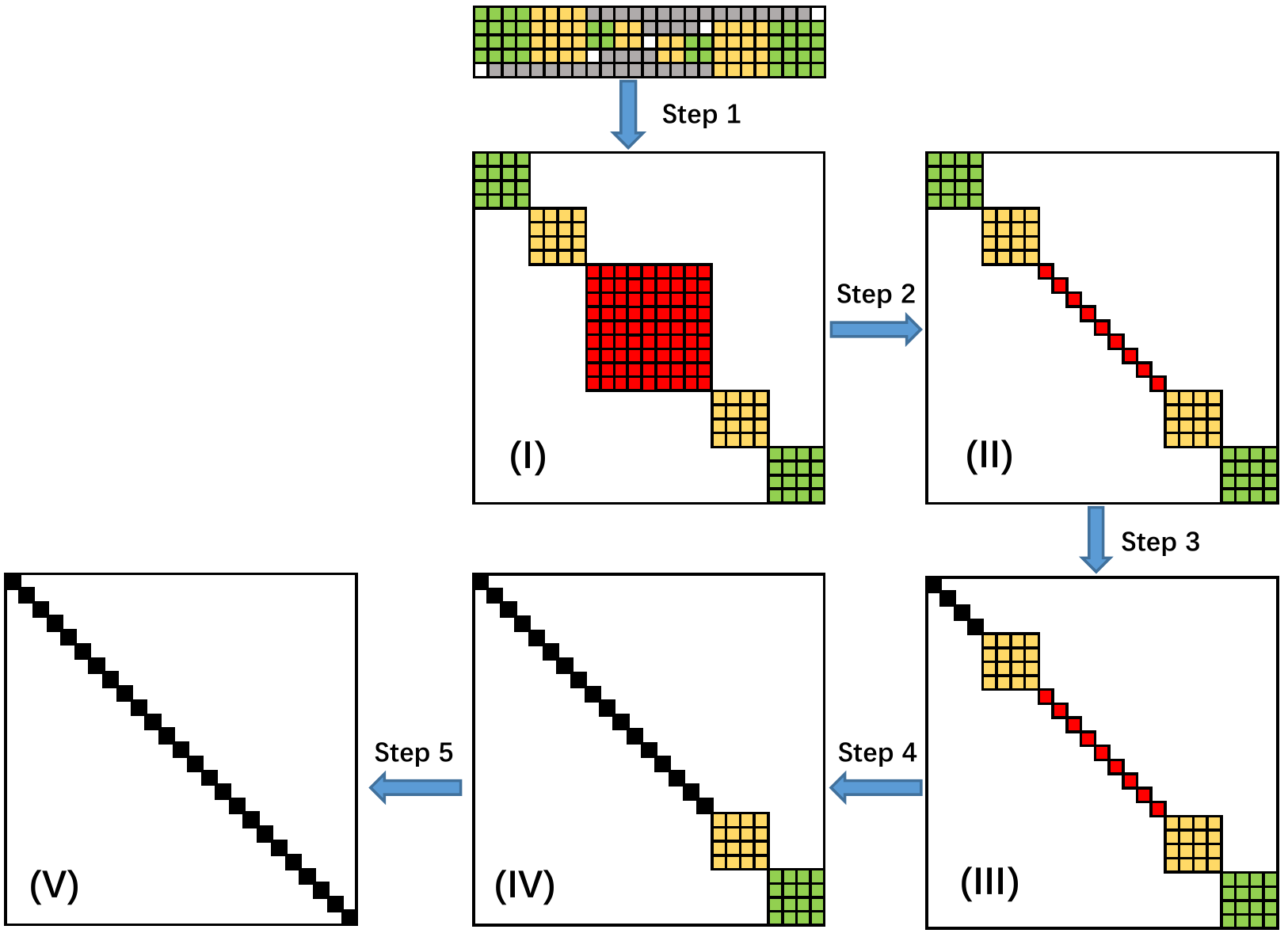}
		\caption{Proving steps for the strongly nonlocal UPBs in $d\otimes d \otimes d$ for $d\geq 5$.  }  \label{fig:titeddd}
	\end{figure}
		
		\noindent{\bf Step 2}  One notice that
		$
		\bra{S_d}\bbI\otimes E\ket{\psi}=\bra{S_{d-2}'}\bbI\otimes E\ket{\psi}
		$
		for any $\ket{\psi}\in\cup_{k=1}^{\delta-1}\cC^{(d,d-2k)}$. Therefore, for any pair of elements $|\phi\rangle$ and $|\psi\rangle$ in  $\{\cup_{k=1}^{\delta-1}\cC^{(d,d-2k)}\cup \{\ket{S_{d-2}'}\}\}$, we have
		$$
		\langle \psi | \bbI' \otimes E_{\cC} |\phi\rangle=\langle \psi | \bbI \otimes E |\phi\rangle=0,
		$$
		where $\bbI':=\sum_{k=1}^{d-2}|k\rangle_A\langle k|.$
		Since
		$\{\cup_{k=1}^{\delta-1}\cC^{(d,d-2k)}\cup \{\ket{S_{d-2}'}\}\}$ is of  the strongest nonlocality, one shows that $$E_{\cC}=\sum_{s=1}^{d-2}\sum_{t=1}^{d-2} L\ket{s}_B\bra{s}\otimes\ket{t}_C\bra{t}$$ for some $L$. Thus the intuitive figure of $E$ can be shown in Fig.~\ref{fig:titeddd} (II).
		
		\noindent{\bf Step 3}  By using the states $\{\ket{\xi_1^{(d)}}_A\ket{0}_B\ket{\eta_i^{(d)}}_C\}_{i\in \bbZ_{d-1}}\subset \cA_1^{(d,d)}$, we have
		\begin{equation}\label{eq:doperatorMu}
			{}_B\bra{0}{}_C\bra{\eta_i^{(d)}}E\ket{0}_B\ket{\eta_j^{(d)}}_C=0, \quad \text{for} \ i\neq j\in\bbZ_{d-1}.
		\end{equation}
		Then there exists a  real number $a_s$  for any $s\in\bbZ_{d-1}$ such that
		\begin{equation}
			E_{(\cA_1^{(d,d)})^{(A)}}=\sum_{s=0}^{d-2}a_s\ket{0}_B\bra{0}\otimes\ket{\eta_s^{(d)}}_C\bra{\eta_s^{(d)}},
		\end{equation}
		as $E=E^{\dagger}$. In the same way, there exist real numbers $a_s,b_s,c_t,e_t$ such that the operator
		\begin{equation}\label{eq:diagonal}
			\begin{aligned}
				E=&\sum_{s=0}^{d-2}a_s\ket{0}_B\bra{0}\otimes\ket{\eta_{s}^{(d)}}_C\bra{\eta_{s}^{(d)}}+\sum_{s=0}^{d-2} b_s\ket{\eta_{s}^{(d)}}_B\bra{\eta_{s}^{(d)}}\otimes\ket{d-1}_C\bra{d-1}+\sum_{t=0}^{d-2} c_t\ket{\xi_{t}^{(d)}}_B\bra{\xi_{t}^{(d)}}\otimes\ket{0}_C\bra{0}\\
				&+\sum_{t=0}^{d-2} e_t\ket{d-1}_B\bra{d-1}\otimes\ket{\xi_{t}^{(d)}}_C\bra{\xi_{t}^{(d)}}+\sum_{s=1}^{d-2}\sum_{t=1}^{d-2} L\ket{s}_B\bra{s}\otimes\ket{t}_C\bra{t}.
			\end{aligned}
		\end{equation}
		By using the states $ \{\ket{0}_A\ket{\eta_i^{(d)}}_B\ket{\xi_j^{(d)}}_C\}_{(i,j)\in\mathbb{Z}_{d-1}\times \mathbb{Z}_{d-1}\setminus \{(0,0)\}}=\cB_3^{(d,d)}$, we can get the following equality
		\begin{equation}\label{MainOrthogonal}
			{}_B\bra{\eta_k^{(d)}} {}_C\bra{\xi_\ell^{(d)}}  E \ket{\eta_i^{(d)}}_B\ket{\xi_j^{(d)}}_C=0, \quad \text{for} \ (k,\ell)\neq (i,j) \in\mathbb{Z}_{d-1}\times \mathbb{Z}_{d-1}\setminus \{(0,0)\}.
		\end{equation}
		We assume that $k\neq i$ and $\ell= j$. Then by Eq.~\eqref{eq:diagonal}, we have
		\begin{equation}\label{eq:mainzero}
			\begin{aligned}
				0=&{}_B\bra{\eta_k^{(d)}} {}_C\bra{\xi_\ell^{(d)}}  E \ket{\eta_i^{(d)}}_B\ket{\xi_\ell^{(d)}}_C\\
				=&\sum_{s=0}^{d-2}a_s   \braket{ \eta_k^{(d)}}{0}_B \braket{0}{\eta_i^{(d)}}_B  \braket{\xi_\ell^{(d)}}{\eta_s^{(d)}}_C \braket{\eta_s^{(d)}}{\xi_\ell^{(d)}}_C + \sum_{s=0}^{d-2}b_s  \braket{\eta_k^{(d)}}{\eta_s^{(d)}}_B\braket{\eta_s^{(d)}}{\eta_i^{(d)}}_B\braket{\xi_\ell^{(d)}}{d-1}_C\braket{d-1}{\xi_\ell^{(d)}}_C\\
				&+\sum_{s=1}^{d-2}\sum_{t=1}^{d-2}L\braket{\eta_k^{(d)}}{s}_B\braket{s}{\eta_i^{(d)}}_B\braket{\xi_\ell^{(d)}}{t}_C\braket{t}{\xi_\ell^{(d)}}_C\\
				=&\sum_{s=0}^{d-2}a_s\braket{\xi_\ell^{(d)}}{\eta_s^{(d)}}_C \braket{\eta_s^{(d)}}{\xi_\ell^{(d)}}_C+\sum_{s=1}^{d-2}\sum_{t=1}^{d-2}Lw_{d-1}^{(i-k)s}\\
				=&\sum_{s=0}^{d-2}a_s\braket{\xi_\ell^{(d)}}{\eta_s^{(d)}}_C \braket{\eta_s^{(d)}}{\xi_\ell^{(d)}}_C-(d-2)L.
			\end{aligned}
		\end{equation}
		There are two cases of the terms in the summation of the last equality.
		\begin{enumerate}[(a)]
			\item If $s= \ell$, then
			\begin{equation*}
				\braket{\xi_\ell^{(d)}}{\eta_s^{(d)}}_C\braket{\eta_s^{(d)}}{\xi_\ell^{(d)}}_C=\sum_{n=1}^{d-2}w_{d-1}^{\ell}\sum_{n=1}^{d-2}w_{d-1}^{-\ell}=(d-2)^2.
			\end{equation*}
			\item If $s\neq \ell$, then
			\begin{equation*}
				\braket{\xi_\ell^{(d)}}{\eta_s^{(d)}}_C\braket{\eta_s^{(d)}}{\xi_\ell^{(d)}}_C=\sum_{n=1}^{d-2}w_{d-1}^{ns-(n-1)\ell}\sum_{n=1}^{d-2}w_{d-1}^{(n-1)\ell-ns}=\sum_{n=1}^{d-2}w_{d-1}^{(n-1)(s-\ell)}\sum_{n=1}^{d-2}w_{d-1}^{(n-1)(\ell-s)}=1.
			\end{equation*}
		\end{enumerate}
		Therefore, Eq.~\eqref{eq:mainzero} is equivalent to
		\begin{equation}\label{eq:holdforall}
			\sum_{s=0}^{d-2}a_i+((d-2)^2-1)a_\ell-(d-2)L=0.
		\end{equation}
		However, the Eq.~\eqref{eq:holdforall} is satisfied for any $\ell\in \bbZ_{d-1}$. Thus, we have  $a_0=a_1=\cdots=a_{d-2}$. It implies that 
		\begin{equation}\label{eq:eA_11}
		E_{(\cA_1^{(d,d)})^{(A)}}=k\bbI_{(\cA_1^{(d,d)})^{(A)}}.
		\end{equation} 
		The intuitive figure of $E$ can be shown in Fig.~\ref{fig:titeddd} (III).

		\noindent{\bf Step 4}
			Considering $\ket{S}$ and $\{\ket{0}_A\ket{\eta_i^{(d)}}_B\ket{\xi_j^{(d)}}_C\}_{(i,j)\in\mathbb{Z}_{d-1}\times \mathbb{Z}_{d-1}\setminus \{(0,0)\}}=\cB_3^{(d,d)}$. By using Eqs.~\eqref{eq:matrix}  and \eqref{eq:eA_11}, we have  the following equality
	\begin{equation}
	\sum_{s=0}^{d-2}\sum_{t=0}^{d-2} {}_B \bra{s} {}_C\bra{t+1} E \ket{\eta_i^{(d)}}_B\ket{\xi_j^{(d)}}_C=\sum_{s=0}^{d-1}\sum_{t=0}^{d-1} {}_B \bra{s} {}_C\bra{t} E \ket{\eta_i^{(d)}}_B\ket{\xi_j^{(d)}}_C=0.
	\end{equation}
	Moreover,  we have
	\begin{equation}
	\sum_{s=0}^{d-2}\sum_{t=0}^{d-2}\ket{s}_B\ket{t+1}_C=\ket{\eta_0^{(d)}}_B\ket{\xi_0^{(d)}}_C.
	\end{equation}
	Therefore, by using the states $\{\ket{S}\}\cup\{\ket{0}_A\ket{\eta_i^{(d)}}_B\ket{\xi_j^{(d)}}_C\}_{(i,j)\in\bbZ_{d-1}\times \bbZ_{d-1}\setminus\{(0,0)\}}$, we have
	\begin{equation}
	{}_B\bra{\eta_k^{(d)}}_C\bra{\xi_\ell^{(d)}}E\ket{\eta_i^{(d)}}_B\ket{\xi_j^{(d)}}_C=0,  \quad  \text{for} \ (k,\ell)\neq (i,j)\in\bbZ_{d-1}\times\bbZ_{d-1}.
	\end{equation}
	
	For any $\ket{t_1}_B\ket{t_2}_C\in (\cA_1^{(d,d)})^{(A)}\cap(\cB_3^{(d,d)})^{(A)}$, we have  ${}_{\{\ket{t_1}_B\ket{t_2}_C\}}E_{(\cB_3^{(d,d)})^{(A)}\setminus \{\ket{t_1}_B\ket{t_2}_C\}}=\textbf{0}$ by Eqs.~\eqref{eq:matrix}  and \eqref{eq:eA_11}. Moreover, ${}_B\bra{t_1}{}_C\braket{t_2}{\eta_i^{(d)}}_B\ket{\xi_j^{(d)}}_C\neq 0$ for $(i,j)\in\bbZ_{d-1}\times \bbZ_{d-1}$.  Applying Lemma~\ref{lem:trivial} to $\{\ket{\eta_i^{(d)}}_B\ket{\xi_j^{(d)}}_C\}_{(i,j)\in\bbZ_{d-1}\times \bbZ_{d-1}}$, we have
	\begin{equation}\label{eq:B31}
	E_{(\cB_3^{(d,d)})^{(A)}}=k_1\bbI_{(\cB_3^{(d,d)})^{(A)}}.
	\end{equation} Since $(\cA_1^{(d,d)})^{(A)}\cap (\cB_3^{(d,d)})^{(A)}\neq \emptyset$, it implies $k=k_1$. Thus, by Eqs.~\eqref{eq:eA_11}  and \eqref{eq:B31}, we obtain
	\begin{equation}
	E_{(\cA_1^{(d,d)})^{(A)}\cup(\cB_3^{(d,d)})^{(A)}}=k\bbI_{(\cA_1^{(d,d)})^{(A)}\cup(\cB_3^{(d,d)})^{(A)}}.
	\end{equation}
	The intuitive figure of $E$ can be shown in Fig.~\ref{fig:titeddd} (IV).

    \noindent{\bf Step 5}	By the symmetry of Fig.~\ref{fig:titedodd}, we can obtain  $E=k\bbI$.  The intuitive figure of $E$ can be shown in Fig.~\ref{fig:titeddd} (V).
    
     Thus,  $E$  is trivial. This completes the proof.

		(ii) The proof is similar as (i).	
	\end{proof}
	
	\vspace{0.4cm}

\end{document}